\documentclass[11pt]{article} 
\usepackage[T1]{fontenc}
\usepackage[utf8]{inputenc}

\usepackage{microtype}

\usepackage{subcaption}
\usepackage{graphicx}
\usepackage[backend=biber,sorting=none,sortcites,citestyle=numeric-comp,bibstyle=phys,doi=false,eprint=true,url=false,biblabel=brackets,isbn=true,articletitle=true]{biblatex}   
\usepackage{authblk}

\usepackage{amsmath, amstext}    
\usepackage{cases}
\usepackage{mathtools}
\usepackage{dsfont}
\usepackage[dvipsnames]{xcolor}
\usepackage[margin=1in]{geometry}
\usepackage{float}
\usepackage{ragged2e}
\DeclareCaptionJustification{justified}{\justifying}
\usepackage{placeins}
\usepackage[title,titletoc]{appendix}
\usepackage[colorlinks=true, linkcolor = BlueViolet, citecolor = BlueViolet, urlcolor = BlueViolet]{hyperref}
\usepackage{enumitem}

\usepackage[bitstream-charter]{mathdesign}

\usepackage{subdepth} 

\usepackage{amsthm}   

\usepackage{moresize}  

\usepackage{makecell}

\allowdisplaybreaks   

\newtheorem{theorem}{Theorem}[section]
\newtheorem{definition}[theorem]{Definition}
\newtheorem{lemma}[theorem]{Lemma}

\newtheorem{corollary}[theorem]{Corollary}

\DeclareMathOperator{\Tr}{Tr}
\DeclareMathOperator{\sinc}{sinc}

\DeclareMathOperator{\vectorspan}{span}

\DeclareMathOperator{\ifc}{\overline{\Theta}}

\newcommand{\uh}{\hat{U}_H}

\newcommand{\proj}{\hat{\Pi}}

\newcommand{\idop}{\hat{\mathds{1}}}
\newcommand{\diff}{\mathrm{d}}
\newcommand{\dos}{\mathcal{N}}   

\newcommand{\kraus}[2]{\hat{\mathcal{K}}_{#1}(#2)}
\newcommand{\krausd}[2]{\hat{\mathcal{K}}_{#1}^\dagger(#2)}

\newcommand{\htime}{{\rm H}}

\newcommand{\effdim}{{\rm eff}}

\newcommand{\obs}{\text{obs}}

\newcommand{\eqm}{\text{eq}}
\newcommand{\neqm}{\text{neq}}

\title{Aperiodicity is sufficient for macroscopic thermalization}
\author[1]{Amit Vikram}
\affil[1]{JILA and Center for Theory of Quantum Matter, Department of Physics, University of Colorado, Boulder CO 80309 USA}
\date{}

\begin{document}
\pagenumbering{gobble}
\maketitle

\abstract{
We identify a general mechanism for the finite-time thermalization of macroscopic observables, such as coarse-grained charge densities, in terms of elementary forms of the quantum dynamics of initial states: (1) aperiodicity, which provides a computable measure of (2) a dynamical \textit{partially} ergodic exploration of the Hilbert space. Specifically, this mechanism predicts the equilibration of all (concentrated) macroscopic observables, in almost all states in an initial ensemble and almost all times within finite and longer intervals, given only the observable-independent information that the return probability of the ensemble of initial states is small over a finite time range. As a special case, it also accesses standard results on equilibration over infinitely long times in terms of (stronger versions of) the effective dimension of initial state delocalization in the energy eigenbasis. Our results incorporate macroscopic thermalization into the domain of operational quantum statistical mechanics, recently developed to provide finitely computable criteria for microscopic thermalization. We discuss an overall characterization of this approach as establishing connections between (1) the decay of a (theoretically or experimentally) computable probe indicating memorylessness, (2) a fundamental invariant mechanism in terms of the alignment of observables or states in the Hilbert space, and (3) predicting different natural forms of (classical and) quantum thermalization, most of which rigorously recover conventional eigenstate-based descriptions of infinite-time thermalization as a special case but provide stronger accessible predictions over finite observation times in the thermodynamic limit.
}

\newpage

\setcounter{tocdepth}{3}
\tableofcontents

\newpage

\pagenumbering{arabic}

\section{Introduction}
\label{sec:intro}

\subsection{Macroscopic thermalization by concentration}
\label{sec:introconcentration}

\subsubsection{Statistical mechanics at different length scales}

Many isolated systems with several degrees of freedom appear to reach an effective stationary or equilibrium state on the measurement of many naturally accessible ``simple'' observables, despite continuing to undergo extremely complex forms of dynamics in their state space. This process of ``thermalization'' is fundamentally characterized by an apparent loss of sensitivity of such simple observables to any detailed information imprinted on the basic degrees of freedom of the system. Where it occurs, thermalization allows the use of equilibrium statistical mechanics~\cite{Tolman, KhinchinStatMech, LLStatMech}: the description of the behavior of such accessible observables by ``universal'' statistical ensembles that are at best privy to extremely coarse grained (if any) information about the state of the system.

Thermalization may be attributed to a variety of distinct overlapping mechanisms across different length scales in various broad classes of systems. At the most microscopic scales, one is concerned with observables pertaining to a few fundamental degrees of freedom, such as the state of a single particle in a many-particle system. Classically, microscopic thermalization in its most elementary form is impossible: the state of any given particle remains well defined and dynamically evolving at all times, and cannot be described at any instant by a coarse-grained statistical ensemble. Nevertheless, weaker forms of statistical behavior can be described in the dynamics of such local observables, e.g. using statistical ensembles to describe their distribution in time along the lines of the \textit{ergodic theory}~\cite{HalmosErgodic, Sinai1976, SinaiCornfeld} (cf.~\cite{KhinchinStatMech}) of classical dynamical systems, restricted to describe specific local observables of interest~\cite{KhinchinStatMech}.

The situation is considerably sharper for microscopic quantum thermalization~\cite{Nandkishore, gogolin2016equilibration, MoriETHreview}. Here, it is possible for the quantum information in a single particle to spread across the system via ``information scrambling''~\cite{HaydenPreskill, ChaosComplexityRMT}, which generates large-scale entanglement and allows the reduced quantum state of a single particle to effectively settle into an equilibrium mixed state, even as the full dynamics in the Hilbert space proceeds unitarily without information loss. With the development of local control and measurement techniques that address individual qubits in experimental quantum platforms, including quantum simulators~\cite{QuantumSimulationReview2024}, it has become possible to observe qubit-wise microscopic quantum thermalization in systems of several (at present $\lesssim 10^2$) qubits~\cite{GoogleThermalization2025, QuantinuumThermalization2026}.
Moreover, such processes remain important to consider (or mitigate) in sophisticated protocols involving the coherent manipulation of quantum information, such as various quantum algorithms~\cite{NielsenChuang}, and even toy models of quantum information loss of relevance to quantum gravity~\cite{HaydenPreskill}.

The statistical mechanics of everyday experience, however, remains largely disconnected from microscopic thermalization. Here, one often observes systems of incomparably large numbers of particles $N$ (e.g. comparable to Avogadro's number~\cite{CODATA2022}, $N \sim 10^{24}$), often justifying a direct $N\to \infty$ limit with all other independent quantities treated as finite, without having anything close to the ability to discern individual particles. The relevant observables tend to be global densities over a (nominally) finite fraction of the characteristic length scales of the system, such as pressure, (kinetic) temperature, or various charge densities (including particle number, a $U(1)$ charge, or magnetization, an $SU(2)$ charge\footnote{Here, the term ``charge'' is used for (Hermitian) generators of ``local'' (per-particle) unitary transformations~\cite{charges}.}).
Understanding the behavior of these observables requires analyzing thermalization over macroscopically large length scales. For this purpose, the scrambling mechanism can be largely irrelevant: macroscopic thermalization prominently occurs even without entanglement generation. This is the case, for example, (1) in classical systems which possess no notion of entanglement, as well as (2) for observables defined across the full system in isolated quantum systems, where there is no external system to scramble information (generate entanglement) with.

A very general phenomenon that contributes to macroscopic thermalization in both classical and quantum systems is the concentration of (probability) measure~\cite{KhinchinStatMech}, a standard occurrence in the theory of probability for aggregations of a large number of variables~\cite{LugosiEtAlConcentration}.
As macroscopic densities usually average a large number of microscopic variables, e.g. the central limit theorem~\cite{RossProbability} (implying concentration) comes into play, and one finds that the distribution of any such observable is very narrow, concentrating around its equilibrium value for large system sizes (a standard example is worked out in Sec.~\ref{sec:concentratedobservables} for completeness). With such macroscopic observables practically being in equilibrium in almost all states in the state space of the system by default, it is then understood to be fairly typical and expected behavior that macroscopic observables thermalize in large systems~\cite{KhinchinStatMech, GallavottiErgodic}.

\subsubsection{Can we predict thermalization dynamics beyond concentration?}
\label{sec:thermalizationbeyondconcentration}

While the measure concentration phenomenon provides a static justification for the \textit{typicality} of macroscopic thermal equilibrium in almost all states of complex systems, the dynamical mechanism of macroscopic thermalization from a nonequilibrium state has not yet been satisfactorily resolved to our knowledge --- in either classical or quantum systems. Such nonequilibrium states are straightforward to prepare and observe in many cases~\cite{gogolin2016equilibration}, and explaining their thermalization is therefore of significant interest.

While it is broadly recognized~\cite{GallavottiErgodic, TasakiTypicalityThermalization} that the dynamics has to be of such a nature as to take a nonequilibrium state into the larger equilibrium region, sharpening this straightforward intuitive idea in terms of precise dynamical notions has been challenging. The primary obstacle to identifying such a mechanism stems from the fact that the state space of a many-particle system is exponentially large in the number of particles $N$, with $\sim \exp(N)$ independent states. In most cases of interest, the equilibrium and nonequilibrium regions of macroscopic observables are both exponentially large in $N$, even with the former occupying nearly the full state space and the latter occupying an exponentially vanishing fraction~\cite{KhinchinStatMech} (as in the example in Sec.~\ref{sec:concentratedobservables}).

This interplay of exponential sizes causes quantitative problems
for characterizing the dynamics of nonequilibrium states in the $N \to \infty$ thermodynamic limit. Dynamical statements in classical ergodic theory generally apply over infinite times, and to \textit{almost all} states rather than all states, corresponding to regions of finite nonzero measure with respect to some probability measure defined on the state space. Without the $N\to\infty$ limit being taken first, ``infinite times'' for finite $N$ are longer than any scale of interest set by $N$, including $\sim 10^{24}$ (Avogadro's number), yielding no predictions for finite times. In the strict $N\to\infty$ limit, the nonequilibrium region is an exponentially small fraction of the full state space, and its \textit{relative size} $\sim \exp(-N)$ shrinks rapidly to measure zero, beyond the reach of ergodic-theory-type statements. However, the absolute size of the nonequilibrium region grows as $\sim \exp(N)$, which is divergent and contains too many states to be tackled on a state-by-state basis. Analogous obstacles exist for quantum systems. We will now illustrate how these obstacles hinder some standard approaches to classical and quantum macroscopic thermalization\footnote{Throughout these arguments and the rest of our developments, we will use the formal asymptotic notation~\cite{knuth1976asymptotic}, by which $f(x) = M(g(x))$ means $[f(x)\ \mathrm{R}\ cg(x)]\ \mathrm{S}(c)$, where the symbol $M$ is (initially) chosen from $M \in (o, O, \Omega, \omega)$, the relation $R$ from $R \in (<, \leq, \geq, >)$ respectively, and the qualifier $\mathrm{S}(c)$ is ``there exists a constant $c$ such that the relation applies'' for the uppercase $O$ and $\Omega$, and ``the relation applies for every constant $c$'' for the lowercase $o$ and $\omega$, with the additional symbol $M = \Theta$ being defined as the simultaneous validity of both uppercase relations $O$ and $\Omega$, when $x\to\infty$ (which may be implicit). Separately, we use $=$ to declare equality, $\equiv$ for a definition/equivalence, and occasionally $:=$ for an \textit{assignment} of a specific variable e.g. in an example, to a general variable that may stand in for several such examples. We also use the boolean indicator function $\ifc(\mathrm{X}) = 1$ if statement $\mathrm{X}$ is true, and $\ifc(\mathrm{X}) = 0$ if statement $\mathrm{X}$ is false, where the overline differentiates the indicator function $\ifc$ from the asymptotic relation $\Theta$.}.

\paragraph{Classical obstacles to predicting macroscopic thermalization}

For classical systems, Khinchin's formal approach~\cite{KhinchinStatMech} argues at length for the concentration mechanism from probability theory being the primary justification for macroscopic statistical mechanics, but also attempts to briefly provide a dynamical justification for the approach to equilibrium in terms of the properties of single-particle observables. For the latter, let $\lbrace a_k\rbrace_{k=1}^{n_O}$ refer to a set of $n_O$ single particle observables, shifted by a constant without loss of generality to have a microcanonical average of $\langle a_k(x)\rangle_{x\in\mathcal{P}} = 0$ over the state space $\mathcal{P}$ (with averages implicitly being with respect to the state space measure). It is shown in Ref.~\cite[Chap. III]{KhinchinStatMech} that given each $\mathcal{P}$-averaged autocorrelation function for this set:
\begin{equation}
C_k(t,t') \equiv \langle a_k(x,t)a_k(x,t')\rangle_{x\in\mathcal{P}},
\end{equation}
its long-time decay\footnote{This is stated in Eq.~\eqref{eq:KhinchinMolecularChaos} as an exact limit following Ref.~\cite{KhinchinStatMech} but note that a weaker C\'{e}saro limit or even just a time averaged decay suffices; see also \cite{dynamicalqthermalization, dynamicalpurestatethermalization}.} characterizing an almost ubiquitous intuitive phenomenon of e.g. ``molecular chaos'' in gases~\cite{KhinchinStatMech},
\begin{equation}
    \lim_{|t-t'| \to \infty}C_k(t,t') = 0,
    \label{eq:KhinchinMolecularChaos}
\end{equation}
establishes the time-averaged thermalization of the corresponding $a_k$ in \textit{almost all} initial states in nonzero measure regions, in a sense similar to (observable-specific) ergodic theory:
\begin{equation}
    \left\langle \left(\lim_{T\to\infty}\int_0^{T}\frac{\diff t}{T}\ a_k(x,t)\right)^2\right\rangle_{x\in\mathcal{P}} = 0.
    \label{eq:KhinchinTimeAveragedThermalization}
\end{equation}
It is then \textit{argued} that all macroscopic observables of the form $\sum_{k=1}^{n_O} c_k a_k$ also average to $0$ almost everywhere in the state space and over infinite time intervals, which indeed follows for any \textit{finite} $n_O$. This meets two key obstacles:
\begin{enumerate}
    \item For any finite $N$, Eqs.~\eqref{eq:KhinchinMolecularChaos} and \eqref{eq:KhinchinTimeAveragedThermalization} already refer to infinite times $t \to \infty$, or stated somewhat more hyperbolically,
    \begin{equation}
        t \sim t' \sim |t-t'| \sim T = \omega(f(N)) \text{ for any } f(N) \neq 0, \text{ e.g. } f(N) = e^{e^{N^2}},
        \label{eq:infinitetimehyperbole}
    \end{equation}
    and infinitely small correlators
\begin{equation}
    \lim_{|t-t'| \to \infty}C_k(t,t') = o\left(f(N)\right) \text{ for any } f(N) \neq 0, \text{ e.g. } f(N) = e^{-e^{N^2}},
\end{equation}
which are respectively too absurdly large and small to apply to the statistical mechanics of everyday experience involving e.g. an Avogadro's number $N \sim 10^{24}$ of particles. In particular, the corresponding timescales $\exp(\exp(\sim10^{48}))$ in some practical units are several exponents beyond the estimated age of the universe $\sim 10^{61}$ \cite{universeage} in units of the Planck time~\cite{CODATA2022}, yielding no predictions for statistical mechanics as observed over the timescale of several seconds to days (or so far in our universe).
\item If the above criteria are formulated instead in the $N\to\infty$ limit, so that softer limits of correlator magnitudes can be taken with more reasonable timescales (e.g. $C_k$ smaller than $\epsilon$ for $|t-t'| > 1/\delta$ and $T > 1/\delta$ for $\epsilon, \delta = \Theta(1)$ even as $N\to\infty$), then the nonequilibrium subspace is already a measure zero set in the full state space (of relative measure $\exp(-\Theta(N)) \to 0$). Therefore, the correlator decay in Eq.~\eqref{eq:KhinchinMolecularChaos} is entirely dominated by equilibrium states, the implication Eq.~\eqref{eq:KhinchinTimeAveragedThermalization} again applies only to equilibrium states (as phase space averages freely exclude sets of measure zero), and the passage to macroscopic densities $\sum_{k=1}^{n_O} c_k a_k$ now has $n_O \to \infty$ (if $n_O = \Theta(N)$, to be macroscopic), with no clear notion of convergence. In short, no prediction is obtained for the thermalization dynamics of macroscopic observables in nonequilibrium states.
\end{enumerate}
The ``molecular chaos'' approach has other merits: in other related work~\cite{dynamicalqthermalization, dynamicalpurestatethermalization}, we have taken the viewpoint that a similar set of mechanisms based on local correlators of specific observables provides a robust computation-friendly justification for \textit{microscopic} quantum thermalization. But it does not appear to us that the classical mechanism gives clear nontrivial predictions in its originally proposed domain of macroscopic thermalization for the aforementioned reasons, whose quantum analogues we will briefly describe in Sec.~\ref{sec:operationalquantumstatmech}.



\paragraph{Quantum macroscopic thermalization over infinite times} For quantum systems, the question of explaining macroscopic thermalization has remained open since von Neumann~\cite{vonNeumannThermalization}, with a series of subsequent works progressing towards its resolution, e.g.~\cite{GoldsteinetalMacroscopicTypicality, GoldsteinetalMateMite1, GoldsteinetalMateMite2, TasakiTypicalityThermalization}. The most up-to-date comprehensive (and accessible) treatment of macroscopic quantum thermalization \textit{from nonequilibrium states} that we are aware of is by Tasaki~\cite{TasakiTypicalityThermalization}. Here, two different dynamical mechanisms are considered for a finite Hilbert space dimension (requiring finite $N$), neither being able to predict finite-time thermalization\footnote{Another set of typicality results show that ``typical'' Hamiltonians thermalize extremely fast~\cite{TypicalFastMacroThermalization1, TypicalFastMacroThermalization2} (with such fast timescales being allowed in every system~\cite{GoldsteinHaraTasakiTimescales}), but this has the opposite problem that almost no realistic system is observed to thermalize as fast~\cite{TasakiTypicalityThermalization}, for which such typicality results give no predictions.} in the $N\to\infty$ limit (a situation highlighted in Ref.~\cite{TasakiTypicalityThermalization} as ``quite unsatisfactory'' and a key open problem). These are:
\begin{enumerate}
    \item Strong eigenstate thermalization away from nonequilibrium: If $\proj_{\neqm}$ is a projector onto the nonequilibrium region, its infinite time average in any state $\lvert \psi(t)\rangle$ (with $\langle \psi\vert \psi\rangle = 1$) is given by an average of its expectation values in the energy eigenstates $\lvert E_n\rangle$, assumed to be nondegenerate (which can be adapted to degeneracies in ways we won't discuss~\cite{TumulkaEsT, dynamicalqthermalization}):
    \begin{equation}
        \lim_{t\to\infty}\int_0^{T}\frac{\diff t}{T} \langle \psi(t)\vert \proj_{\neqm}\lvert \psi(t)\rangle = \sum_{n} \left\lvert\langle E_n\vert \psi\rangle \right\rvert^2 \langle E_n\rvert \proj_{\neqm}\lvert E_n\rangle.
        \label{eq:Tasaki_infytimeaverage}
    \end{equation}
    Therefore, if each energy eigenstate individually has a low probability e.g. at most some $\epsilon > 0$ of being found in the nonequilibrium region (i.e. thermalizes ``away from nonequilibrium''),
    \begin{equation}
        \langle E_n\rvert \proj_{\neqm}\lvert E_n\rangle \leq \epsilon,
        \label{eq:strongETH}
    \end{equation}
    then the infinite time average in Eq.~\eqref{eq:Tasaki_infytimeaverage} is necessarily not greater than $\epsilon$, showing that any state spends at most an $\epsilon$ fraction of its infinite-time history in the nonequilibrium region (as witnessed by projective measurements~\cite{NielsenChuang}). This should give macroscopic thermalization over almost all times. While this mechanism is advocated for by a number of works e.g. \cite{GoldsteinetalMateMite1, GoldsteinetalMateMite2, TumulkaEsT}, similar issues to the classical case remain: (a) Eq.~\eqref{eq:Tasaki_infytimeaverage} is still an infinite time statement for finite $N$ and is susceptible to the same issue as Eq.~\eqref{eq:infinitetimehyperbole}, while (b) strong eigenstate thermalization~\eqref{eq:strongETH} itself involves exponentially many eigenstates (Hilbert space size) in $N$, and may not be computationally verifiable in general with any finite resolution for thermodynamically large systems~\cite{ThermalizationUndecidability1, ThermalizationUndecidability2}.
    \item Energy eigenstate delocalization of initial states: Using the Cauchy-Schwarz inequality~\cite{ByronFuller} on Eq.~\eqref{eq:Tasaki_infytimeaverage} (and further bounding a resulting factor involving $\proj_{\neqm}$), we get
    \begin{equation}
        \lim_{t\to\infty}\int_0^{T}\frac{\diff t}{T} \langle \psi(t)\vert \proj_{\neqm}\lvert \psi(t)\rangle \leq \sqrt{\left(\sum_n\left\lvert\langle E_n\vert \psi\rangle \right\rvert^4\right)\Tr[\proj_{\neqm}]}.
        \label{eq:Tasaki_effdim1}
    \end{equation}
    Noting that $\Tr[\proj_{\neqm}]$ is the size of the nonequilibrium region, if the inverse participation ratio $\sum_n \lvert\langle E_n\vert \psi\rangle \rvert^4$ (which estimates the inverse of the number of energy eigenstates $\lvert \psi\rangle$ is significantly supported on) satisfies $\sum_n \lvert\langle E_n\vert \psi\rangle \rvert^4 \ll 1/\Tr[\proj_{\neqm}]$, then the time average on the left hand side is small, again showing that the state spends most of its time outside the nonequilibrium region. As with Eq.~\eqref{eq:strongETH}, comparable issues resurface: (a) the infinite time average gives no predictions for finite times as $N\to\infty$ as with Eq.~\eqref{eq:infinitetimehyperbole}, and (b) the exponential size of the nonequilibrium region in $N$ prevents computationally constraining the participation ratio of $\lvert \psi\rangle$ to be as large with finite resolution in the thermodynamic limit.
\end{enumerate}

\paragraph{A preview of our results} The goal of this work is then to identify a mechanism that can allow a computational prediction of macroscopic quantum thermalization over finite timescales from finite-resolution data: in short, allowing predictions in principle in the observable regime of macroscopic statistical mechanics. Our main result is Theorem~\ref{thm:aperiodicityimpliesmacroscopicthermalization}, which does not require any limits, but allows taking $N\to\infty$ first if desired while retaining finite timescales and finite magnitudes of measurable quantities. In particular, we show that if an ensemble of nonequilibrium states of sufficient size can be prepared (the precise experimental protocols for which we leave to future work), a finite suppression of its return probability over a finite time interval is sufficient to establish macroscopic thermalization at almost all times over that interval as well as arbitrarily long time intervals, in almost all states in this nonequilibrium ensemble. Up to the ability to prepare or emulate this ensemble (as discussed in Sec.~\ref{sec:disc_computability}), a measurement of such a return probability can be readily carried out with finite computational resources, making this a robust operational criterion for macroscopic thermalization. The underlying mechanism is a quantum dynamical notion of \textit{partial ergodicity}, extending an operational notion in Ref.~\cite{dynamicalqentanglement}, connected to exploring orthogonal states in the Hilbert space~\cite{dynamicalqergodicity, dynamicalqentanglement}. This partial form of ergodicity provides a precise formulation of the idea that an initial nonequilibrium state must explore enough of the state space to reach the equilibrium region \textit{and} be guaranteed to remain there at most times, where conventional ergodic theory lacks the resolution to make finite-time predictions due to its insistence on exploring the full state space. The next subsection places this work in a more specific context, involving parallel motivations in quantum dynamics and statistical mechanics in addition to the explanation of macroscopic thermalization.


\subsection{Parallel motivations and summary}

\subsubsection{Operational quantum statistical mechanics}
\label{sec:operationalquantumstatmech}

Quantum statistical mechanics can be described at different formal levels. The most general level is \textit{typicality}~\cite{vonNeumannThermalization, tumulka_CT, CanonicalTypicalityPSW, NormalTypicality}, which applies a statistical approach to explaining quantum statistical mechanics itself. One looks at the behavior of uniform ensembles (with respect to the Haar measure) of states in, or observables acting on, the Hilbert space, and finds that almost all of them generally thermalize with (Haar) measure zero exceptions~\cite{vonNeumannThermalization, tumulka_CT, CanonicalTypicalityPSW, NormalTypicality}. Typicality is (largely, up to excluding measure-zero degenerate configurations of the energy spectrum for dynamics~\cite{vonNeumannThermalization, NormalTypicality}) system-independent, and applies to all quantum systems irrespective of their observed behavior, even while not providing an explanation of the latter. This is because any discrete set of states (or observables) of physical interest, such as an orthonormal basis of the Hilbert space (which can be prepared by projective measurements~\cite{NielsenChuang}) or any countable collection of them, has a Haar measure of exactly zero, 
and therefore behaves entirely independently of typicality.

A more refined layer is provided by that of energy eigenstate structure hypotheses, which posit the behavior of observables~\cite{JensenShankarETH, deutsch1991eth, srednicki1994eth, srednicki1999eth, rigol2008eth, DAlessio2016, deutsch2018eth} or states~\cite{Tasaki1998, ReimannRealistic, LindenetalEqb, ShortEqb, ShortDegenerate} of physical interest in the energy eigenbasis, focusing on individual energy eigenstates. This is closely related to the typicality layer: these hypotheses are essentially that such observables or states look typical in the energy eigenbasis (specifying the basis-specific structure missing from the Haar measure) in ``sufficiently complex'' systems, and therefore thermalize, with more refinements. Specifically, the former refers to the eigenstate thermalization hypothesis for ``local'' observables~\cite{JensenShankarETH, deutsch1991eth, srednicki1994eth, srednicki1999eth, rigol2008eth, DAlessio2016, deutsch2018eth}, and the latter to energy eigenstate delocalization for ``generic'' initial states~\cite{Tasaki1998, ReimannRealistic, LindenetalEqb, ShortEqb, ShortDegenerate}. By the energy time uncertainty principle, $\Delta E \cdot \Delta t \gtrsim 1$ (crudely invoked here as a Fourier relation, without using more precise speed limit formulations~\cite{MT, AAMT, ML, LevitinToffoli, GongHamazakiNoneqbBounds}), the focus on individual energy eigenstates $\Delta E \to 0$ leads to a divergent time uncertainty $\Delta t \to \infty$. Correspondingly, these approaches only constrain dynamics over infinite times, which as noted above (e.g. Eq.~\eqref{eq:infinitetimehyperbole}) cannot make predictions for any finite range of times especially in the thermodynamic limit. Questions of finite-time observable statistical mechanics remain elusive in this layer, as has been extensively recognized for both observable-based and state-based eigenstate structure hypothesis in more formal discussions~\cite{TasakiTypicalityThermalization, ShortDegenerate, GarciaPintosFiniteTimeEqb, dynamicalqthermalization}.

In recent works~\cite{dynamicalqthermalization, dynamicalpurestatethermalization}, it has been our goal to develop a third layer: the prediction of thermalization over finite times from a finite amount of accessible data that can be computed in various ways --- theoretically (analytical calculations), numerically (classical simulations), or experimentally (quantum simulations). A key feature of this approach is to avoid imposing any explicit limits (such as the thermodynamic limit or infinite time limit), so that the predictions apply quite generally and any desired limits may be taken after obtaining the relevant prediction. The requisite input data varies with the kind of thermalization one is interested in. One convenient organizing principle~\cite{dynamicalpurestatethermalization} is to separate different types of thermalization by the type of averaging involved. This originates in the fact that any local (microscopic) observable $\hat{a}$ usually has eigenvalues quite distinct from the thermal value (e.g. the average of eigenvalues), and one therefore needs different kinds of averaging over measurement outcomes (which necessarily lie among the eigenvalues~\cite{DiracQM, ShankarQM}) to obtain the thermal value. This is detailed in the discussion below, summarized in Table~\ref{tab:operationalstatmech}.

\begin{table}[!t]
    \centering
    \centerline{
    \begin{tabular}{|c|c|c|c|}
    \hline
        \makecell{\textbf{Type of averaging} \\ \textbf{(and thermalization)}} & \textbf{Classical limit?} & \textbf{Mechanism type} & \textbf{ Computable diagnostic }  \\
        \hline
         \makecell{Spatial / Densities \\ (Macroscopic)} & Yes & State-dependent &  Aperiodicity [this work] \\
         \hline
         \makecell{Time / Ergodic \\ (Microscopic)} & Yes & Observable-dependent & Time-averaged autocorrelator~\cite{dynamicalqthermalization}\\
         \hline
         \makecell{Initial states / Mixing \\ (Microscopic)} & Yes & Observable-dependent & Autocorrelator relaxation~\cite{dynamicalqthermalization}\\
         \hline
         \makecell{Quantum expectation value \\ (Microscopic)} & No & \makecell{Observable-dependent \\ (Mildly state-dependent)} & Out-of-time-ordered correlator~\cite{dynamicalpurestatethermalization} \\
         \hline
    \end{tabular}}
    \caption{Operational quantum statistical mechanics, classified by the relevant type of averaging. From top to bottom, these types are sorted by decreasing relevance to the statistical mechanics of everyday experience, and increasing relevance to the technical arena of manipulating quantum information.}
    \label{tab:operationalstatmech}
\end{table}

\paragraph{Microscopic thermalization with averages} Our earlier developments in this layer remained within the domain of microscopic statistical mechanics, characterizing each observable by a computable dynamical property of that observable~\cite{dynamicalqthermalization, dynamicalpurestatethermalization}. Averaging a microscopic observable over time (for a given initial state) or over a large class of initial states (for a given time) is possible for both classical and quantum systems, and correspond to ergodic- and mixing- like behavior~\cite{HalmosErgodic, Sinai1976, SinaiCornfeld, KhinchinStatMech} for a given observable.  For these cases, the decay of the (classical or quantum) autocorrelator\footnote{For energy-dependent thermal values, this class of autocorrelators include complex-valued ``echoes'' with different durations of time evolution on each side of the operator.} of $\hat{a}$ over a fixed interval of time implies thermalization in almost all (classes of) initial states in any orthonormal basis over all longer timescales~\cite{dynamicalqthermalization} (see also a concise technical description in \cite{dynamicalpurestatethermalization}). This mechanism recovers the microscopic ``molecular chaos'' mechanism in Eqs.~\eqref{eq:KhinchinMolecularChaos} and \eqref{eq:KhinchinTimeAveragedThermalization} in the case of time averages under the classical limit, the thermodynamic limit,  the limit of infinitely long times, and the limit of infinitesimally small autocorrelators at these times, taken in whichever order (with the order stated here being the appropriate one to recover the classical formulation in Eqs.~\eqref{eq:KhinchinMolecularChaos} and \eqref{eq:KhinchinTimeAveragedThermalization} in their stated setting). More significantly, it also predicts weak versions of the diagonal and off-diagonal eigenstate thermalization hypothesis for $\hat{a}$, including for \textit{approximate} energy eigenstates spanning a finite energy range $\Delta E$, recovering the relevant parts of the eigenstate structure layer as special cases with $\Delta E \to 0$ while being compatible with thermalization over a finite timescale $T \gtrsim 1/\Delta E$.

\paragraph{Microscopic quantum thermalization} Allowing no averages other than the quantum expectation value $\langle \psi(t)\vert \hat{a}\lvert \psi(t)\rangle$ in the state of interest has no classical analogue, but such thermalization in almost all states in any orthonormal basis (for a large subspace) can be predicted by computing a distinctly quantum few-body out-of-time-ordered correlator~\cite{dynamicalpurestatethermalization}. This mechanism also involves a mild state-dependence as the out-of-time-ordered correlator is that of the observable and a large subspace of states, but this only requires specifying the state on a (nominally) finite number of qubits even in a thermodynamically large system, leaving the state of all other qubits quite general and arbitrary. All these microscopic correlators are computed in a single mixed state --- a maximally mixed state, which can be purified~\cite{NielsenChuang} by entanglement with a larger auxiliary system (that does not participate in the dynamics) to allow a single correlator in this single quantum state to determine the thermalization behavior of the microscopic observable in almost all initial states compatible with the subspace (i.e. fixed on a few qubits), over finite times.

\paragraph{Computability} It is worth emphasizing that the aforementioned observable-dependent computable thermalization statements apply to almost all states in an orthonormal basis. This is considerably stronger than typicality results in accessing a large fraction of states that can be experimentally prepared (e.g. a projective measurement projects onto a single orthonormal basis~\cite{NielsenChuang}), but stops short of predicting thermalization in all states. Here, it is also worth noting that a computational mode of predicting thermalization in all states may not be possible in the thermodynamic limit, according to analyses of the decidability of thermalization based on Turing machines~\cite{ThermalizationUndecidability1, ThermalizationUndecidability2}. Crucially, we may be fundamentally restricted to operational predictions only for ``almost all'' states by constraints of computability, which must be taken into account for macroscopic thermalization. An immediate consequence for our present purposes is that the observable-dependent strong eigenstate thermalization mechanism for macroscopic thermalization in Eqs.~\eqref{eq:Tasaki_infytimeaverage} and \eqref{eq:strongETH} likely cannot be made computationally accessible in general (in addition to its infinite time limitation), and is therefore not a suitable starting point for an operational approach to this form of thermalization.

\paragraph{Macroscopic thermalization} The most straightforward average --- over a large number of different particles or spatial locations --- takes us from microscopic to macroscopic statistical mechanics, which is the subject of this work. Here, the primary challenge is that, as noted in Sec.~\ref{sec:introconcentration}, almost all states are in thermal equilibrium by default, and the interesting class of nonequilibrium initial states forms an exponentially small fraction of the state space (i.e. the Hilbert space). Therefore, the observable-dependent statements from previous works~\cite{dynamicalqthermalization, dynamicalpurestatethermalization} that constrain thermalization in almost all states (with the excluded states forming a small but \textit{finite} fraction for computation with finite resolution) do not provide any nontrivial predictions for macroscopic thermalization (at finite times) for these nonequilibrium states. Here, in contrast, noting that concentration is already a strong property of macroscopic observables that we essentially get for free, we will seek computable properties of the nonequilibrium states themselves that can decide finite-time thermalization in a large fraction of these states. In other words, we will adopt a state-dependent approach, and show that a relevant property of nonequilibrium initial states that establishes finite-time macroscopic thermalization is \textit{aperiodicity} (specified by Definition~\ref{def:aperiodicity}).

\subsubsection{State-dependent equilibration bounds}
\label{sec:statedependenteqbounds}

Having implicated state-dependent mechanisms as being of more relevance for (computable predictions of) macroscopic thermalization, it is worth reviewing previously known statements of this type. The primary statement of relevance to macroscopic thermalization is Eq.~\eqref{eq:Tasaki_effdim1}, adapted from Ref.~\cite{TasakiTypicalityThermalization}, which contains a measure of how delocalized the initial state is over energies. Formally, the initial state delocalization measure~\cite{Tasaki1998, ReimannRealistic, LindenetalEqb, ShortEqb, ShortDegenerate} can be expressed in terms of an ``effective dimension'' in the energy eigenbasis,
\begin{equation}
    d_{\effdim} = \frac{1}{\left\lvert\langle E_n\vert \psi\rangle \right\rvert^4},
    \label{eq:intro_effdim_def}
\end{equation}
estimating a representative number of energy eigenstates over which $\lvert \psi\rangle$ has significant support. As noted around Eq. (8), $d_{\effdim} \gg 1$ is sufficient to guarantee some degree of macroscopic thermalization over infinite times. In addition, the same effective dimension has general implications for \textit{equilibration} at all length scales, whether the equilibrium value is the thermal value or not. Specifically, $d_{\effdim} \gg 1$ also guarantees that the expectation value of any observable remains close to its long-time average at almost all times, over an \textit{infinite} time interval~\cite{Tasaki1998, ReimannRealistic, LindenetalEqb, ShortEqb, ShortDegenerate}.

However, as with eigenstate thermalization, it does not appear possible to make sharper statements over time intervals that are finite in the thermodynamic limit. The sharpest state-dependent (and observable-independent) bounds so far can restrict the time interval to somewhat longer than the inverse scale of consecutive energy level spacings~\cite{ShortDegenerate}. General stronger bounds on the equilibration time appear impossible as e.g. there exist observables that can only equilibrate over times almost as long as the inverse scale of energy level spacings in states that are highly delocalized in the energy eigenbasis (such as projector observables onto the optimal cyclic permutation basis of \cite{dynamicalqergodicity} in an initial basis state that is unbiased across all energy eigenstates, whose equilibration is limited by spectral rigidity). Once again, it appears that the fundamental obstacle is the energy-time uncertainty principle, where delocalization over states with $\Delta E = 0$ (or more generally $\Delta E = \delta E$, where $\delta E$ is the typical scale of consecutive energy level spacings) can only constrain thermalization dynamics over $T \to \infty$ (or more generally $T \gtrsim 1/\delta E$).

Here, where macroscopic thermalization is concerned, we will show that aperiodicity (Definition~\ref{def:aperiodicity}) can access finite thermalization timescales, unlike the effective dimension in Eq.~\eqref{eq:intro_effdim_def}, as it demonstrably contains much more information about eigenstate structure over larger energy scales. In particular, aperiodicity over a finite timescale $T$ implies the delocalization of the initial state not only over different energy eigenstates, but also over different larger energy shells of width $\Delta E \sim 1/T$ (as shown in Corollary~\ref{cor:energyshelldelocalization}). In this way, our approach retains compatibility with the energy-time uncertainty principle, while being smoothly connected to eigenstate delocalization results (especially Eq.~\eqref{eq:Tasaki_effdim1}) for $\Delta E = 0$ and $T \to \infty$. It may be interesting to analyze whether this state-dependent criterion also has nontrivial implications for the finite-time equilibration of more general non-macroscopic observables (without requiring strong computationally expensive observable-specific assumptions as in Ref.~\cite{GarciaPintosFiniteTimeEqb}, which nevertheless tends to require highly mixed initial states), which we leave for future work.



\subsubsection{The role of ergodic dynamics in thermalization}
\label{sec:ergodicitythermalization}

It is sometimes customary to invoke ergodic dynamics in discussions of statistical mechanics~\cite{LLStatMech, Reif, ReichlStatMech}, mostly on intuitive grounds. A lesser known tradition with more formal mathematical backing recognizes that ergodic dynamics generally makes no predictions for finite-time statistical mechanics~\cite{KhinchinStatMech, GallavottiErgodic, BricmontErgodicityCritique, GoldsteinErgodicityCritique, LebowitzAPS2021}, relying instead on typicality arguments as realized by the concentration mechanism to motivate the latter in the macroscopic regime. As noted in Sec.~\ref{sec:introconcentration} and also highlighted in Ref.~\cite{TasakiTypicalityThermalization}, even typicality is not sufficient to make finite-time predictions in this regime. We will show that a certain degree of \textit{partially ergodic dynamics}, precisely characterized in this work (Sec.~\ref{sec:theorem} and Sec.~\ref{sec:proof}), makes typicality work for observable macroscopic thermalization at finite times.

It is instructive to trace the flow of reasoning in the statements above pertaining to ergodicity and typicality (incidentally, both notions are due to Boltzmann~\cite{GallavottiErgodic, GoldsteinErgodicityCritique}) for more intuition. We have already argued in Sec.~\ref{sec:operationalquantumstatmech} that a state-dependent approach may be necessary for macroscopic thermalization. Ergodicity, strictly defined~\cite{HalmosErgodic, Sinai1976, SinaiCornfeld}, refers to almost all states in the state space exploring (any small neighborhood of) every state over infinitely long times. A more natural restriction in light of state-dependence is to demand the ergodicity of a single state, i.e. that the infinite time average of every observable in this state is thermal. A finite time restriction is not possible: due to the exponentially large size of the state space, such an exploration requires an exponentially long time, and has a similar order of limits issue as Eq.~\eqref{eq:infinitetimehyperbole}. This cannot be resolved merely by restricting ergodicity to observables of interest~\cite{KhinchinStatMech}, the primary context of Eq.~\eqref{eq:infinitetimehyperbole}.

\begin{figure}[!t]
    \centering
    \includegraphics[width=0.8\linewidth]{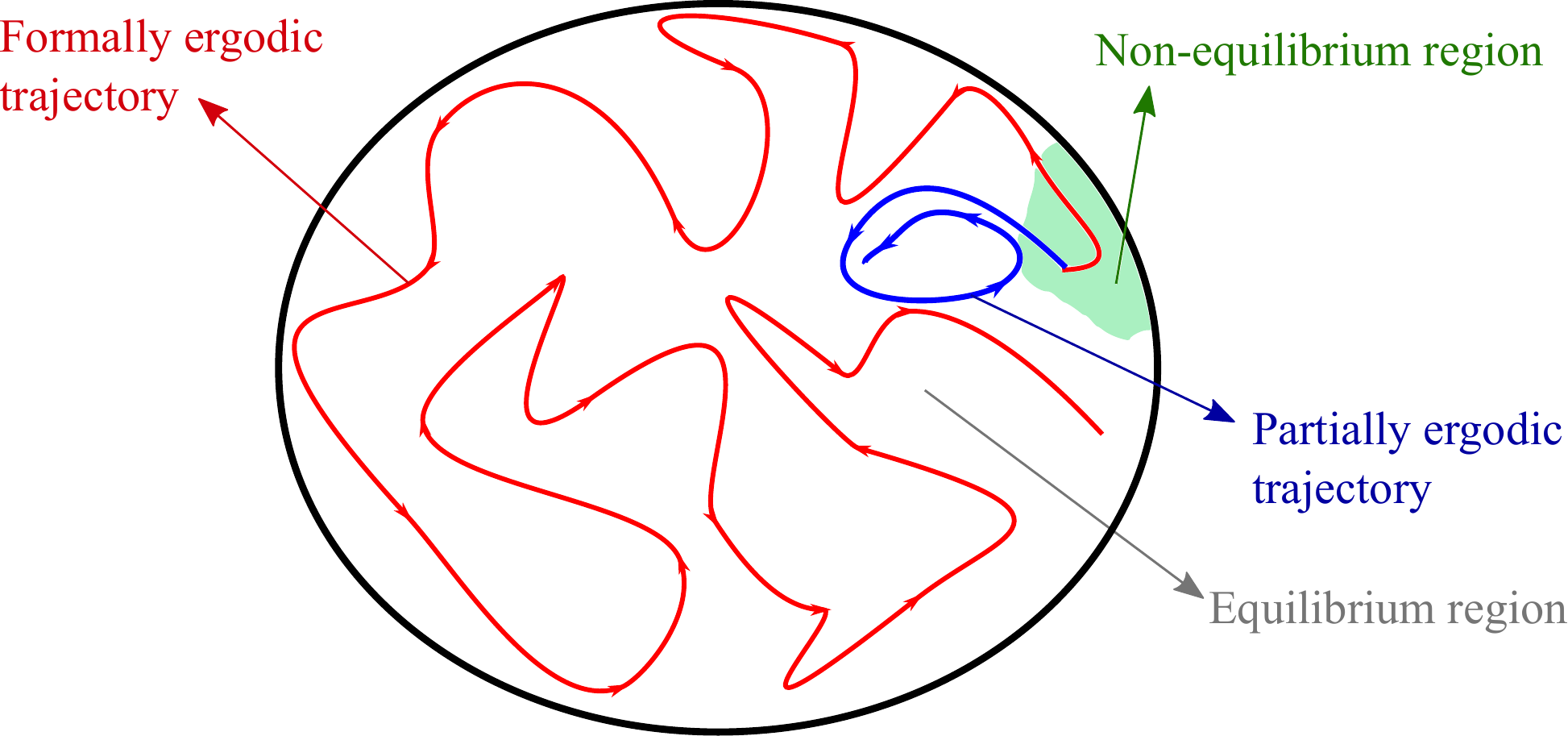}
    \caption{A depiction of the role of ergodic dynamics in macroscopic thermalization in a schematic state space. The equilibrium region takes up almost all of the phase space, and the nonequilibrium region (shaded) takes up a vanishing fraction (pictured as disproportionately large and connected here for convenience). Partial ergodicity (exploring a fraction of the phase space) is sufficient for macroscopic thermalization at almost all times along the trajectory, without requiring formal ergodicity (exploring the full phase space). Moreover, partial ergodicity is necessary to establish finite-time macroscopic thermalization, which conventional ergodicity cannot formally guarantee (for reasons not apparent in this depiction: formal ergodicity allows the state to remain entirely within the nonequilibrium region within finite times; see the text for details).}
    \label{fig:partialergodicity}
\end{figure}

Typicality is widely regarded as a satisfactory solution~\cite{KhinchinStatMech, BricmontErgodicityCritique, GoldsteinErgodicityCritique, LebowitzAPS2021}, but as noted in Sec.~\ref{sec:operationalquantumstatmech}, the typicality here is with respect to a measure that takes appreciable values over such a large space (the state space, e.g. the classical phase space or the quantum Hilbert space) in the thermodynamic limit that it cannot access nonequilibrium states of interest. Interestingly, this limitation does not appear to have been given due importance in analyses that point out the limitations of ergodicity but emphasize the explanatory nature of typicality~\cite{GoldsteinErgodicityCritique, LebowitzAPS2021}, but is somewhat indirectly more transparent in the quantum case~\cite{TasakiTypicalityThermalization}, via the considerations described in Sec.~\ref{sec:thermalizationbeyondconcentration}. If available, a simple resolution would be to determine a criterion that certifies when an initial state evolves into the typical region, as depicted in Fig.~\ref{fig:partialergodicity} (see Ref.~\cite{TasakiTypicalityThermalization} for similar depictions).

We have found it more natural to approach such a criterion directly in quantum systems, where a form of ergodic dynamics associated with exploring new orthogonal directions in the Hilbert space can be identified~\cite{dynamicalqergodicity, dynamicalqentanglement}. More significantly for our purposes, the number of new orthogonal directions explored can be quantified via continuous measures~\cite{dynamicalqentanglement} even without exploring the full Hilbert space, in particular for only a finite interval of time (or even a finite number of discrete time steps). This allows precisely formulating a notion of how much of the Hilbert space needs to be explored to reach a thermal equilibrium subspace of known size, i.e. partial ergodicity\footnote{Partial ergodicity as discussed here is a state-dependent property over finite intervals of time, and is distinct from the decomposition of the state space into ergodic subsets that may be forced by conservation laws, which just amounts to conventional ergodicity (or metric indecomposability~\cite{HalmosErgodic, Sinai1976, SinaiCornfeld, KhinchinStatMech}) over infinite times within subspaces of the state space.} as we will develop in Sec.~\ref{sec:theorem} and Sec.~\ref{sec:proof}, giving a finite-time thermalization bound in Theorem~\ref{thm:aperiodicityimpliesmacroscopicthermalization} (with additional finite-time implications related to partial ergodicity in Sec.~\ref{sec:corollaries}). While we do not explicitly address the classical limit of quantum systems that admit one here, it appears likely that such a notion of partial ergodicity may be formulated classically  (in terms of alternative measures sensitive to nonequilibrium regions) without relying on the (quantum~\cite{DiracQM, ShankarQM} or classical~\cite{Koopman_KvN, vonNeumann_KvN1, vonNeumann_KvN2, HalmosErgodic, Sinai1976, SinaiCornfeld}) Hilbert space via the connection~\cite{dynamicalqergodicity} between this notion of quantum ergodicity and permutations of discrete cells in a classical phase space, which can closely approximate classical dynamics~\cite{Sinai1976, SinaiCornfeld, KatokStepin1, KatokStepin2, KatokSinaiStepin, Nadkarni}. Moreover, a heuristic classical version at the level of classical permutation dynamics on a discretized phase space is discussed in Sec.~\ref{sec:classicalpermutations}, to provide intuition for the quantum mechanism in Theorem~\ref{thm:aperiodicityimpliesmacroscopicthermalization}.

\subsection{Organization of this paper}

Sec.~\ref{sec:theorem} (Theorem) motivates and states our main result, Theorem~\ref{thm:aperiodicityimpliesmacroscopicthermalization}, which is proved with a detailed account of the physical mechanism underlying the theorem in Sec.~\ref{sec:proof} (Proof). Subsequently, Sec.~\ref{sec:corollaries} (Corollaries) discusses a number of consequences of the theorem or steps in its proof, including (1) the ability to establish the delocalization of the initial state across energy shells from computable data in Corollary~\ref{cor:energyshelldelocalization}, recovering the effective dimension of interest to the infinite time state-dependent equilibration bounds discussed in Sec.~\ref{sec:statedependenteqbounds}, (2) the problem of long timescales in deciding thermalization for a pure state in Corollary~\ref{cor:purestatethermalization} (consistent with the construction of a slow thermalizing subspace for any state in Ref.~\cite{GoldsteinHaraTasakiTimescales}), and (3) the use of Theorem~\ref{thm:aperiodicityimpliesmacroscopicthermalization} to establish the finite-time thermalization of several observables for an overwhelming fraction of an ensemble of states, via Corollary~\ref{cor:multipleobservablefinitetimethermalization}. Sec.~\ref{sec:discussion} (Discussion) is more semi-quantitative: it discusses connections to the other approaches of Table~\ref{tab:operationalstatmech} suggesting a general formal structure in the computational prediction of thermalization, explores some questions of computability in practice, and analyzes the general role of (partially) ergodic dynamics across different instances of quantum dynamics and quantum statistical mechanics.

\section{Theorem}
\label{sec:theorem}

To formalize our setup, we work in a general finite-dimensional Hilbert space $\mathcal{H}$ of dimension $D = \dim \mathcal{H} > 1$, with time evolution for a state $\lvert \psi\rangle \in \mathcal{H}$ generated by a Hermitian Hamiltonian operator $\hat{H}$, according to~\cite{DiracQM, ShankarQM}
\begin{equation}
    \lvert \psi(t)\rangle \equiv e^{-i \hat{H} t} \lvert \psi\rangle.
    \label{eq:purestateevolution}
\end{equation}
In this section, we motivate an operational criterion for establishing the equilibration of a large class of observables in a large class of initial state in $\mathcal{H}$ in terms of a computationally verifiable dynamical property of the latter, formulated in Theorem~\ref{thm:aperiodicityimpliesmacroscopicthermalization}, with a detailed proof in Sec.~\ref{sec:proof}. Its direct physical implications (corollaries) are discussed in Sec.~\ref{sec:corollaries}.

\subsection{Concentrated observables and equilibrium subspaces}
\label{sec:concentratedobservables}

Our main result concerns observables that have large \textit{equilibrium subspaces}: subspaces of the Hilbert space in which the observable is nearly constant (i.e. approximately acts as the identity operator, up to scaling). This is a standard notion going back to Ref.~\cite{vonNeumannThermalization} and notably revisited in different forms in e.g. Refs.~\cite{NormalTypicality, GoldsteinetalMacroscopicTypicality, GoldsteinetalMateMite1, GoldsteinetalMateMite2, TasakiTypicalityThermalization}. We will formally define such a notion for a single observable $\hat{A}$, with a chosen equilibrium value $A_{\eqm}$ and resolution $\epsilon > 0$ (unlike Refs.~\cite{vonNeumannThermalization, NormalTypicality, GoldsteinetalMateMite1, GoldsteinetalMateMite2, TasakiTypicalityThermalization}, we will not need to consider more than one observable at a time at this stage). With $\hat{A}$ being Hermitian, and $\lbrace a_k\rbrace_{k=1}^{D}$ its (real valued) eigenvalues with respective eigenvectors $\lbrace \lvert a_k\rangle\rbrace_{k=1}^{D}$, we take the equilibrium subspace to be specified by:
\begin{definition}[Equilibrium subspace of an observable]
\label{def:equilibriumsubspace}
    For a Hermitian observable $\hat{A}$, the equilibrium subspace $\mathcal{H}_{\eqm}(\hat{A}, A_{\eqm}, \epsilon) \subseteq \mathcal{H}$ with equilibrium value $A_{\eqm}$ and resolution $\epsilon$ is the subspace spanned by the eigenvalues of $\hat{A}$ which are at most a distance $\epsilon$ away from $A_{\eqm}$:
    \begin{equation}
        \mathcal{H}_{\eqm}(\hat{A}, A_{\eqm}, \epsilon) \equiv \vectorspan\left\lbrace \lvert a_k\rangle: \left\lvert a_k - A_{\eqm}\right\rvert \leq \epsilon\right\rbrace.
    \end{equation}
    It is also convenient to define the associated equilibrium projector $\proj_{\eqm}(\hat{A}, A_{\eqm}, \epsilon)$, identified by
    \begin{equation}
       \proj_{\eqm}(\hat{A}, A_{\eqm}, \epsilon)\lvert \psi\rangle = 1 \iff \lvert \psi\rangle \in \mathcal{H}_{\eqm}(\hat{A}, A_{\eqm}, \epsilon) 
    \end{equation}
    with $\proj_{\eqm}^2(\hat{A}, A_{\eqm}, \epsilon) = \proj_{\eqm}(\hat{A}, A_{\eqm}, \epsilon)$ and $\proj_{\eqm}^\dagger(\hat{A}, A_{\eqm}, \epsilon) = \proj_{\eqm}(\hat{A}, A_{\eqm}, \epsilon)$. The dimension of the equilibrium subspace is then given by $D_{\eqm}(\hat{A}, A_{\eqm}, \epsilon) \equiv \dim \mathcal{H}_{\eqm}(\hat{A}, A_{\eqm}, \epsilon) = \Tr[\proj_{\eqm}(\hat{A}, A_{\eqm}, \epsilon)]$.
\end{definition}
Note that as an operator, $\proj_{\eqm}(\hat{A}, A_{\eqm}, \epsilon)$ gives the probability that a projective measurement~\cite{NielsenChuang} of $\hat{A}$ gives an outcome $a$ within an $\epsilon$-distance of the equilibrium value $A_{\eqm}$:
\begin{equation}
    \mathbb{P}(a \in [A_{\eqm}-\epsilon, A_{\eqm}+\epsilon] \text{ in } \lvert \psi\rangle \in \mathcal{H}) = \langle \psi\rvert \proj_{\eqm}(\hat{A}, A_{\eqm}, \epsilon)\lvert \psi\rangle.
\end{equation}
It is also convenient to have a corresponding formal notion~\cite{TasakiTypicalityThermalization} of a nonequilibrium subspace\footnote{Stated simply, $\mathcal{H}_{\neqm}$ is spanned by the remaining eigenvectors of $\hat{A}$ not in $\mathcal{H}_{\eqm}$.} that represents the complement (in terms of outcomes) of the equilibrium subspace:
\begin{definition}[Nonequilibrium subspace of an observable]
    For a given $\mathcal{H}_{\eqm}(\hat{A}, A_{\eqm}, \epsilon)$, which is a subspace of $\mathcal{H}$, find the unique decomposition of the latter into a complete pair of orthogonal subspaces $\mathcal{H}_{\eqm}(\hat{A}, A_{\eqm}, \epsilon) \perp \mathcal{H}_{\neqm}(\hat{A}, A_{\eqm}, \epsilon)$:
    \begin{equation}
        \mathcal{H} = \mathcal{H}_{\eqm}(\hat{A}, A_{\eqm}, \epsilon) \oplus \mathcal{H}_{\neqm}(\hat{A}, A_{\eqm}, \epsilon).
    \end{equation}
    Then $\mathcal{H}_{\neqm}(\hat{A}, A_{\eqm}, \epsilon)$ is the nonequilibrium subspace corresponding to the triplet $(\hat{A}, A_{\eqm}, \epsilon)$, with projector $\proj_{\neqm}(\hat{A}, A_{\eqm}, \epsilon) = \idop - \proj_{\eqm}(\hat{A}, A_{\eqm}, \epsilon)$ and dimension $D_{\neqm}(\hat{A}, A_{\eqm}, \epsilon) = D - D_{\eqm}(\hat{A}, A_{\eqm}, \epsilon)$.
\end{definition}

Note that for a given triplet $(\hat{A}, A_{\eqm}, \epsilon)$, the equilibrium subspace can even be empty, or (nearly) as large as the full Hilbert space $\mathcal{H}$. Two examples illustrate general cases of intuitive interest:
\begin{enumerate}
    \item \textbf{Local operators}: In a many-body system of qubits, let $\hat{A} := \hat{\sigma}_z(\vec{r})$, a Pauli spin operator of a qubit at location $\vec{r}$. We have $\Tr[\hat{\sigma}_z] = 0$ suggesting the natural equilibrium value $A_{\eqm} := 0$, while the eigenvalues of $\hat{\sigma}_z$ are $\lbrace -1, +1\rbrace$. It follows that the equilibrium subspace is empty for $\epsilon < 1/2$, e.g.
    \begin{equation}
        \mathcal{H}_{\eqm}\left(\hat{\sigma}_z(\vec{r}), 0, \frac{1}{2+\delta}\right) = \left\lbrace\right\rbrace \text{ and } D_{\eqm}\left(\hat{\sigma}_z(\vec{r}), 0, \frac{1}{2+\delta}\right) = 0 \text{ for } \delta > 0.
    \end{equation}
    Here, the nonequilibrium subspace is the full Hilbert space.
    \item \textbf{Macroscopic (charge) densities}: Again, in a many body system of $N$ qubits ($D = 2^N$), let $\mathcal{L}$ denote a sublattice of many qubits numbering $L = |\mathcal{L}|$, and define the macroscopic $z$-direction spin density (noting that spin is an $SU(2)$ charge~\cite{charges}) in the sublattice:
    \begin{equation}
        \hat{S}_z(\mathcal{L}) \equiv \frac{1}{L}\sum_{\vec{r}\in\mathcal{L}} \hat{\sigma}_z(\vec{r}).
        \label{eq:spinobservable}
    \end{equation}
    Then, $\hat{S}_z(\mathcal{L})$ has eigenvalues given by all possible sums of $L$ unbiased random variables taking values in $\lbrace -1, +1\rbrace$, which follows the binomial distribution~\cite{RossProbability} with the eigenvalue $s_{zj}(\mathcal{L}) = -1+2j/L$ for $j\in \mathbb{Z}_{L+1}$ being degenerate with multiplicity:
    \begin{equation}
        P_j = \binom{L}{j} 2^{N-L},
    \end{equation}
    where $\binom{n}{k} = n!/(k!(n-k)!)$ is the binomial coefficient.
    Loosely speaking, by the central limit theorem~\cite{RossProbability}, most of these eigenvalues are concentrated inside an $O(1/\sqrt{L})$ distance of $0$, suggesting $A_{\eqm} := 0$. More precisely, for a finite resolution $\epsilon \in (0, 1/2)$, we have
    \begin{equation}
        D_{\eqm}(\hat{S}_z(\mathcal{L}), 0, \epsilon) = 2^{N}\left(1 - 2^{-L+1}\sum_{j=0}^{\lfloor L(1-\epsilon)/2\rfloor} P_j\right),
        \label{eq:binomialeqmsubspacedim}
    \end{equation}
    where $\lfloor x\rfloor$ is the integer floor of $x$. From Bernstein's concentration bound~\cite{LugosiEtAlConcentration, HaarBook, VershyninProbability} for sums of independent random variables (the eigenvalues of $\hat{\sigma}_z(\vec{r})$), we get
    \begin{equation}
        2^{-L}\left(2\sum_{j=0}^{\lfloor L(1-\epsilon)/2\rfloor} P_j\right)\leq 2 \exp\left[-\frac{L \epsilon^2}{2(1+\epsilon/3)}\right],
        \label{eq:BernsteinBound}
    \end{equation}
    and therefore,
    \begin{equation}
        D_{\eqm}(\hat{S}_z(\mathcal{L}), 0, \epsilon) \geq D\left(1-2 e^{-\frac{L \epsilon^2}{2(1+\epsilon/3)}}\right).
        \label{eq:macroscopicdensityconcentration}
    \end{equation}
    For $L \gg 1$, we see that $D_{\eqm}/D \approx 1$, showing that the equilibrium subspace spans nearly the full Hilbert space for large $L$. Stated differently, the nonequilibrium subspace is a small \textit{fraction} of the Hilbert space for large $L$,
    \begin{equation}
        D_{\neqm}(\hat{S}_z(\mathcal{L}), 0, \epsilon) \leq 2^{N+1}e^{-\frac{L \epsilon^2}{2(1+\epsilon/3)}}.
        \label{eq:nonequilibrium_dim_example}
    \end{equation}
    But note that it is entirely possible that the nonequilibrium subspace is still large i.e. $D_{\neqm} \gg 1$ even though $D_{\neqm}/D \approx 0$ for $L \gg 1$. Indeed, for $L=O(1)$ [even if large] in the $N\to\infty$ limit, we have both $D_{\neqm} = \Theta(D)$ and $D_{\neqm}/D \ll 1$ for $L \gg 1$. More precisely, from Eq.~\eqref{eq:binomialeqmsubspacedim}, we have the crude bound:
    \begin{equation}
        D_{\neqm}\ \geq\ 2^{N-L+1} \binom{L}{\lfloor L(1-\epsilon)/2\rfloor}\ \geq\ 2^{N-L+1} \left(\frac{2}{1-\epsilon}\right)^{\lfloor L(1-\epsilon)/2\rfloor},
        \label{eq:densities_Dneqm_bound}
    \end{equation}
    where we have used $\binom{n}{k} \geq (n/k)^k$ and $y/\lfloor yx\rfloor \geq 1/x$ for $x,y > 0$. Thus, for both $L=O(1)$ and $L=\Theta(N)$, we very generically have $D_{\neqm} = 2^{\Theta(N)}$ and $D_{\eqm} \sim 2^{N}$: both the nonequilibrium and equilibrium subspaces are exponentially large in $N$, though the former is much smaller than the latter for $L \gg 1$. Along similar lines, for systems with global charge conservation, $[\hat{H}, \hat{S}_z(\mathcal{L}_{\text{full}})] = 0$ with $\mathcal{L}_{\text{full}}$ being the full set of qubits ($|\mathcal{L}_{\text{full}}| = N$), we expect that thermalization only occurs for initial states with a narrow width in the eigenvalues $\hat{S}_z(\mathcal{L}_{\text{full}})$. For such cases, we may choose to restrict the overall Hilbert space $\mathcal{H}$ of interest to be composed only of eigenspaces (so that $\mathcal{H}$ is closed under time evolution, which remains unitary) corresponding to such a narrow eigenvalue width. 
\end{enumerate}

As qualitative terminology, for a given $\epsilon \ll 1$, if there exists some $A_{\eqm}$ such that $D_{\eqm}(\hat{A}, A_{\eqm}, \epsilon)/D \approx 1$ in some sense, we will call $\hat{A}$ a concentrated observable. The most important class of concentrated observables are macroscopic densities, as illustrated in Eq.~\eqref{eq:macroscopicdensityconcentration}.

If an observable is concentrated, it is \textit{sufficient} for an initial state to explore a small fraction of the Hilbert space (i.e. much larger than $D_{\neqm}/D$) for it to attain equilibrium at almost all times. The question of necessity is more subtle, and as we will see, such exploration is necessary to guarantee the equilibration of all conceivable concentrated observables in the Hilbert space, even if not required for the equilibration of a given observable (e.g. if the state is already in its equilibrium subspace). The technical goal of our main theorem is to formalize this intuition in terms of computationally verifiable notions of quantum dynamics.

\subsection{Intuition from classical permutations}
\label{sec:classicalpermutations}

To gain some insight for our task, it is helpful to consider a classical caricature of quantum dynamics in the Hilbert space that captures its essential features. A convenient model involves permutations of an orthonormal basis for $\mathcal{H}$. These have been previously shown to provide an informative class of reference processes for observable-independent ergodic quantum dynamics with strong connections to spectral statistics~\cite{dynamicalqergodicity}, allowing a clarification of the ``quantum chaos conjecture''~\cite{BerryTabor, DKchaos, CGV, BerryStadium, BGS, Haake}, but we will motivate this model for a more general context here. Further, within a purely classical context, such a model of permutations of objects may also provide a reasonable picture of thermalization in systems with a continuous phase space, if the permuted objects are taken to be phase space cells; classical measure-preserving dynamics can often be well approximated by such permutations given sufficiently many cells~\cite{KatokStepin1, KatokStepin2, KatokSinaiStepin, Sinai1976, SinaiCornfeld, Nadkarni}.

Given that we are working in a finite dimensional Hilbert space, the primary distinction between classical and quantum dynamics that is useful to enforce here is that the former may not generate superpositions of ``classical states''. Taking these states to be a specific orthonormal basis $\mathcal{C} = \lbrace \lvert C_k\rangle\rbrace_{k=0}^{D-1}$ for $\mathcal{H}$, the dynamical operations of interest must map $\mathcal{C} \to \mathcal{C}$. As unitary dynamics is reversible, we want our caricature to capture this reversibility by being invertible, which restricts the allowed operations to be permutations of the objects in $\mathcal{C}$. Further, Hamiltonian dynamics is autonomous, and we will therefore take the classical dynamics to be generated by successive applications of the \textit{same} permutation. In Ref.~\cite{dynamicalqergodicity}, it was shown that a strong notion of quantum ergodicity can be captured by further restricting the class of permutations to cyclic permutations (which make each basis state explore the entire basis on successive applications), but we will not make this restriction here to capture \textit{partial} forms of ergodicity that involve exploring a \textit{fraction} of the Hilbert space as suggested by the intuition described at the end of Sec.~\ref{sec:concentratedobservables}.

In general, every permutation can be decomposed into a set of disjoint cycles~\cite{ByronFuller}. Consider an initial state $\lvert C_0\rangle$ that participates in a $D_C$-element cycle:
\begin{equation}
\ldots \to \lvert C_0\rangle \to \lvert C_{M(1)}\rangle \to \lvert C_{M(2)}\rangle \to \ldots \to \lvert C_{M(D_C-1)}\rangle \to \lvert C_0\rangle \to \ldots,
\end{equation}
where $M(j)$ for $j\in \mathbb{Z}$ labels the successive indices of the cycle with $M(0) = 0$ (and unique values for $j \in \mathbb{Z}_{D_C}$). Let us label this cycle containing $\lvert C_0\rangle$ as $\mathcal{C}_0$. If $D_C = D$, then $\mathcal{C}_0$ is fully ergodic~\cite{dynamicalqergodicity}; however, we are particular interested in cases where it may only be possible to establish that $D_C \gg 1$ (even if $D_C < D$), corresponding (in general) to partially ergodic permutation dynamics.

Now let us consider equilibrium subspaces that are spanned by the $\lvert C_k\rangle$. For definiteness, consider a single (arbitrary) equilibrium subspace $\mathcal{H}_{\eqm}$, with corresponding nonequilibrium subspace $\mathcal{H}_{\neqm}$. If $D_C > D_{\neqm}$, then it is clear that $(D_C-D_{\neqm})$ elements of $\mathcal{C}_0$ are necessarily in $\mathcal{H}_{\eqm}$. Consequently, if $\lvert C_0\rangle$ explores more of the Hilbert space than $D_{\neqm}$, as determined by the size $D_C$ of its cycle, then its trajectory remains in $\mathcal{H}_{\eqm}$ for a fraction of steps given by:
\begin{equation}
    f_{\eqm} = \frac{D_C - D_{\neqm}}{D_C}.
    \label{eq:classicalthermalizationfraction}
\end{equation}
This applies for \textit{any} equilibrium subspace with dimension $D_{\eqm} > D-D_C$. We consequently get equilibration for all concentrated observables of equilibrium dimension $D_{\eqm}$ at almost all times (or discrete steps) with $f_{\eqm} \approx 1$ for $D_C \gg D-D_{\eqm}$.

Another question is how one might operationally determine the size $D_C$ of the cycle $\mathcal{C}_0$. Here, we consider the return probability of $\lvert C_0\rangle$ at the step $j$:
\begin{equation}
    p_0(j) \equiv \ifc\left(\lvert C_{M(j)}\rangle = \lvert C_0\rangle\right),
\end{equation}
and note that its time average satisfies:
\begin{equation}
    \lim_{\tau \to \infty} \frac{1}{\tau}\sum_{j=0}^{\tau - 1} p_0(j) = \frac{1}{D_C}.
\end{equation}
More generally, for any finite time average over $\tau$ steps, we have
\begin{equation}
    D_C \geq \frac{1}{\frac{1}{\tau}\sum_{j=0}^{\tau - 1} p_0(j)},
    \label{eq:classicalergodicityreturnprobability}
\end{equation}
where it is convenient to refer to the quantity as the right hand side as the ``aperiodicity'' of the initial state over $\tau$ time steps (the reciprocal of the time averaged return probability).
This gives, from Eq.~\eqref{eq:classicalthermalizationfraction},
\begin{equation}
    f_{\eqm} \geq 1-\frac{(D-D_{\eqm})}{\tau}\sum_{j=0}^{\tau - 1} p_0(j),
    \label{eq:classicalthermalizationfraction2}
\end{equation}
for the fraction of times spent by $\lvert C_0\rangle$ in the equilibrium subspace of any observable with equilibrium dimension $D_{\eqm}$.
The main takeaway is the following\footnote{An extension is possible to the equilibration of almost all states in terms of the aperiodicity of an initial collection of states, but this requires some disproportionately technical discussion for classical permutations without necessarily improving intuition for our fundamental mechanism. We therefore prefer to eschew a discussion of the general classical case in the text in favor of directly formulating our main quantum result for this more general situation.}: Eq.~\eqref{eq:classicalthermalizationfraction} constrains the fraction of times spent in equilibrium in terms of a partially ergodic exploration of the Hilbert space, Eq.~\eqref{eq:classicalergodicityreturnprobability} constrains the dimension of ergodic exploration in terms of the return probability of the initial state over finite times (an observable quantity in principle), and Eq.~\eqref{eq:classicalthermalizationfraction2} combines the two to constrain the equilibrium fraction directly in terms of the return probability over finite times. In particular, if the return probability remains low (aperiodicity is high), then equilibration is guaranteed at almost all times.

We will now try to obtain a similar structure in quantum dynamics, where superpositions are allowed, finding a similar mechanism of partial ergodicity at play and a similar bound on equilibration in terms of observable return probabilities (albeit quantitatively weaker, with the second term on the right hand side being replaced by the square root of its analogue, to account for superpositions).

\subsection{Aperiodicity \texorpdfstring{$\implies$}{implies} macroscopic quantum thermalization}

For quantum dynamics, our result will take the following form: a state that doesn't return to itself for a finite interval of time (aperiodicity) must necessary explore a large fraction of the Hilbert space (partial ergodicity), and therefore necessarily spend most of its history in the equilibrium subspaces of sufficiently concentrated observables (macroscopic thermalization).

However, were we to restrict such considerations only to one pure state at a time, we would be faced with a number of operational issues. The first is that the operational criterion of aperiodicity will have to be tested one at a time for each pure state, or at least each pure state that is not already in equilibrium for a given set of observables of interest. As nonequilibrium subspaces for macroscopic observables themselves tend to be exponentially large in the (sub)system size [see Eq.~\eqref{eq:nonequilibrium_dim_example}], such verification can be exponentially expensive. The second is that it can take an individual pure state a long amount of time (divergent in the thermodynamic limit) to explore enough of the Hilbert space to necessarily exit every nonequilibrium region (as supported in Sec.~\ref{sec:corollaries_purestate}). In general, we would not obtain a guarantee of thermalization within finite times for a single pure state.

Both issues can be addressed by working with a subspace of initial states, $\mathcal{H}_I \subseteq \mathcal{H}$, with projector $\proj_I$ and dimension $D_I = \dim \mathcal{H}_I = \Tr[\proj_I]$. The projector $\proj_I$ evolves in time to project onto the time-evolved versions of the same set of states; from Eq.~\eqref{eq:purestateevolution}, this means that
\begin{equation}
    \proj_I(t) \equiv e^{-i\hat{H} t} \proj_I e^{i\hat{H} t}.
    \label{eq:projectorevolution}
\end{equation}
We will show that if $\proj_I$ does not return to itself within a finite interval of times (subspace aperiodicity), then almost all states \textit{in every orthonormal basis}\footnote{We note that establishing thermalization for almost all states in every orthonormal basis is a much stronger criterion than standard typicality results~\cite{vonNeumannThermalization, tumulka_CT, CanonicalTypicalityPSW, NormalTypicality}, which establish thermalization for almost all states in the full Hilbert space according to the Haar measure, with insufficient resolution to decide thermalization in any given orthonormal basis (as any discrete set of states has Haar measure zero).} for $\mathcal{H}_I$ necessarily attain thermal equilibrium for all concentrated observables with sufficiently large equilibrium subspaces within finite times (macroscopic thermalization almost everywhere in every basis) on account of the collection of almost all states exploring the rest of the Hilbert space (subspace partial ergodicity).

Let us now make the notion of aperiodicity of a subspace within a finite time interval precise. A good measure of ``(quasi-)periodicity'' is given by the return probability of the projector,
\begin{equation}
    p_I(t,t') \equiv \frac{1}{D_I}\Tr[\proj_I(t)\proj_I(t')].
\end{equation}
For formal reasons, we write this in terms of two time arguments $(t,t')$, but it naturally follows from Eq.~\eqref{eq:projectorevolution} that $p_I(t,t') = p_I(t-t')$ needs to depend only on a single argument. It is also worth noting that $p_I(t,t') = p_I(t',t)$, $p_I(t+t_1,t'+t_1) = p_I(t,t')$, and $p_I(t,t) = 1$. Specifically, we note that the time average of $p_I(t,t')$ over both arguments provides a measure of how frequently states in $\mathcal{H}_I$ ``return'' to have a nonzero overlap with $\mathcal{H}_I$. We define the aperiodicity of the subspace as the reciprocal of this average (where $\diff t$ represents a suitable discrete or continuous integration measure):
\begin{definition}[Aperiodicity of a subspace in a specified range of times]
\label{def:aperiodicity}
    Let $\mathcal{H}_I \subseteq \mathcal{H}$ be a subspace of dimension $D_I$. Its degree of aperiodicity over a set of times $\mathcal{T}$ of duration
    \begin{equation}
        |\mathcal{T}| \equiv \int_{\mathcal{T}} \diff t\ 
    \end{equation}
    is given by:
    \begin{equation}
        \mathcal{A}_{\mathcal{T}}[\mathcal{H}_I] = \frac{1}{\displaystyle{\int_{\mathcal{T}}\frac{\diff t_1}{|\mathcal{T}|}\int_{\mathcal{T}}\frac{\diff t_2}{|\mathcal{T}|}\ p_I(t_1,t_2)}}.
        \label{eq:aperiodicitydef}
    \end{equation}
\end{definition}

For a given observable $\hat{A}$, if $\mathcal{H}_I \subseteq \mathcal{H}_{\neqm}$ lies within its nonequilibrium subspace $\mathcal{H}_{\neqm}$, then it is clear that a large amount of aperiodicity is necessary for $\mathcal{H}_I$ to evolve to have a large overlap with the equilibrium subspace $\mathcal{H}_{\eqm}$. This is because a large return probability $p_I(t_1,t_2) \geq 1-\delta_p$ implies a correspondingly large overlap with $\mathcal{H}_{\neqm}$. In contrast, if $\mathcal{H}_I \subseteq \mathcal{H}_{\eqm}$ consists of equilibrium states with respect to $\hat{A}$ to begin with, no aperiodicity or any kind of dynamics is necessary: the subspace can even remain in equilibrium just by not evolving. However, for a given $\mathcal{H}_I$, as we can always choose some observable $\hat{A}$ (e.g. by applying a unitary transformation) such that $\mathcal{H}_I \subseteq \mathcal{H}_{\neqm}$ for this observable if $D_{\neqm} \geq D_I$, it follows that aperiodicity is necessary for the equilibration of \textit{all conceivable} concentrated observables that can be defined on the Hilbert space, whether or not such observables are accessible in practice.

With this setup, we can state our main result concerning the \textit{sufficiency} of aperiodicity for macroscopic thermalization. For anchoring to classical intuition, the two statements of the theorem, as expressed in Eq.~\eqref{eq:ergodicityfromaperiodicity} and \eqref{eq:equilibrationfractionfromaperiodicity}, are to be compared with Eq.~\eqref{eq:classicalergodicityreturnprobability} and \eqref{eq:classicalthermalizationfraction2} respectively.
\begin{theorem}[Finite-time aperiodicity implies long-time macroscopic thermalization via partial ergodicity]
\label{thm:aperiodicityimpliesmacroscopicthermalization}
    Let $\mathcal{H}_I$ be a subspace of initial states of dimension $D_I$ with projector $\proj_I$, and $\mathcal{T}$ be a set of times of interest. Then, the corresponding aperiodicity measure $\mathcal{A}_{\mathcal{T}}[\mathcal{H}_I]$ determines the following properties:
    \begin{enumerate}
        \item \underline{Partially ergodic exploration of the Hilbert space}: Let $\mathcal{H}_{M}$ be a subspace of dimension $D_M$ with projector $\proj_M$. If the average probability of measuring the history of the subspace $\mathcal{H}_I$, given by the ensemble of states $\lbrace \proj_I(t), t\in\mathcal{T}\rbrace$, to fall within $\mathcal{H}_M$ (measurement outcome $1$ for $\proj_M$) is:
        \begin{equation}
            p_{M}[\mathcal{H}_I, \mathcal{T}] = \int_{\mathcal{T}}\frac{\diff t}{|\mathcal{T}|}\ \frac{1}{D_I}\Tr[\proj_I(t) \proj_M],
        \end{equation}
        then the measurement subspace $\mathcal{H}_M$ must be at least as large as:
        \begin{equation}
            D_M \geq  p_{M}[\mathcal{H}_I, \mathcal{T}]\ D_I \mathcal{A}_{\mathcal{T}}[\mathcal{H}_I].
            \label{eq:ergodicityfromaperiodicity}
        \end{equation}
        Specifically, only subspaces $\mathcal{H}_M$ of size $D_M \gtrapprox D_I \mathcal{A}_{\mathcal{T}}[\mathcal{H}_I]$ can contain most ($p_{M} \approx 1$) of the history of $\mathcal{H}_I$ over the times $\mathcal{T}$, as witnessed by projective measurements (by $\proj_M$). Thus, aperiodicity measures the ``number of dimensions'' of the Hilbert space explored by dynamics over $\mathcal{T}$, nominally given by $D_I \mathcal{A}_{\mathcal{T}}[\mathcal{H}_I]$, in an operational sense.
        
        \item \underline{Equilibration of all concentrated observables in almost all states at almost all times}: Given the subspace $\mathcal{H}_I$ of initial states of interest, let $\mathcal{B}_I = \lbrace \lvert \psi_{Ik}\rangle\rbrace_{k=1}^{D_I} \subset \mathcal{H}_I$ be an orthonormal basis for the subspace $\mathcal{H}_I$. If $\mathcal{T} = [0,T]$ is an interval and $R \in \mathbb{N}$ is any number of ``repetitions'', define the expanded interval $\mathcal{T}_R = [0,RT]$. Then for any concentrated observable with equilibrium subspace $\mathcal{H}_{\eqm}$ of dimension $D_{\eqm}$ with projector $\proj_{\eqm}$, the average fraction of times in $\mathcal{T}_R$ spent by the states in $\mathcal{B}_I$ with an overlap of at least some $(1-\varepsilon) \in [0,1)$ with the equilibrium subspace,
        \begin{equation}
            f_{\eqm}[\mathcal{B}_I, \mathcal{H}_{\eqm}](RT, \varepsilon) \equiv \frac{1}{D_I}\sum_{k=1}^{D_I}\int_{0}^{RT}\frac{\diff t}{RT} \ifc\left(\langle \psi_{Ik}(t)\rvert \proj_{\eqm}\lvert \psi_{Ik}(t)\rangle \geq 1-\varepsilon\right),
            \label{eq:equilibrationfractiondef}
        \end{equation}
        is constrained by the aperiodicity of the subspace $\mathcal{A}_{[0,T]}[\mathcal{H}_I]$ within the original time interval $[0,T]$ (note that $\mathcal{H}_I = \vectorspan\mathcal{B}_I$ is uniquely specified by $\mathcal{B}_I$):
        \begin{equation}
            f_{\eqm}[\mathcal{B}_I, \mathcal{H}_{\eqm}](RT, \varepsilon) \geq \left[1-\frac{1}{\varepsilon}\sqrt{\frac{D - D_{\eqm}}{D_I \mathcal{A}_{[0,T]}[\mathcal{H}_I]}}\right].
            \label{eq:equilibrationfractionfromaperiodicity}
        \end{equation}
        Thus, aperiodicity establishes that almost all states in every orthonormal basis for $\mathcal{H}_I$ spend almost all times in arbitrarily long intervals (including $R\to\infty$) close to thermal equilibrium for \textit{every} observable with sufficiently small nonequilibrium subspace compared to the size of the explored Hilbert space in the finite time interval $[0,T]$, i.e. with $(D-D_{\eqm}) = D_{\neqm} \ll D_I \mathcal{A}_{[0,T]}$.
    \end{enumerate}
\end{theorem}
\begin{proof}[Proof synopsis]
    See Sec.~\ref{sec:proof} for the full proof, which is built from Lemmas~\ref{lem:ergodicfraction}, \ref{lem:ergodicprobesubspace}, \ref{lem:ergodicityimpliesequilibrationfinitetime}, and \ref{lem:finitetimeergodicitytoinfinitetimethermalization}. These establish a partially ergodic exploration of the Hilbert space resulting from aperiodicity as the underlying mechanism for the connection between aperiodicity and macroscopic thermalization. Their combination also has somewhat more general implications than stated in Theorem~\ref{thm:aperiodicityimpliesmacroscopicthermalization}, some of which will be analyzed in Sec.~\ref{sec:corollaries}, with a specific formulation for discrete times discussed in Sec.~\ref{sec:disc_computability}.
\end{proof}

\section{Proof}
\label{sec:proof}

\subsection{Aperiodicity \texorpdfstring{$\iff$}{iff} partially ergodic dynamics}
\label{sec:aperiodicityandergodicity}

Let us review our setup for convenience. We have a Hilbert space $\mathcal{H}$ of dimension $\dim \mathcal{H} = D > 1$. Further, $\proj_I$ projects onto a subspace $\mathcal{H}_I \subseteq \mathcal{H}$ of initial states, of dimension $\dim \mathcal{H}_I = \Tr[\proj_I] = D_I$. As a function of time $t$, given autonomous unitary dynamics with a Hamiltonian $\hat{H}$, its time evolution is given by:
\begin{equation}
    \proj_I(t) \equiv e^{-i\hat{H}t} \proj_I e^{i\hat{H} t}.
\end{equation}
Over time, this evolution traces out an ensemble of states in $\mathcal{H}$, given by the Krylov set that tracks the history~\cite{Sinai1976} of the initial operator:
\begin{equation}
    \mathcal{K}(\proj_I) = \lbrace \proj_I(t)\rbrace_t.
\end{equation}
We would like to focus on a specific ``profile'' of relevance of each element of the Krylov set over time, determined by a weight function $w(t) \geq 0$ satisfying:
\begin{equation}
    \int\diff t\ w(t) = 1.
    \label{eq:w_normalization}
\end{equation}
Assigning a weight $w(t)$ to the time $t$ generates an ensemble from the Krylov set corresponding to the density operator:
\begin{equation}
    \hat{\rho}_w(\proj_I) = \frac{1}{D_I}\int\diff t\ w(t) \proj_I(t).
    \label{eq:KrylovEnsembledef}
\end{equation}
In the simplest case, the weight function is uniformly nonzero inside an interval $t\in[0,T]$ and zero outside, corresponding to a standard time average of $\proj_I$ over $[0,T]$.

Following Ref.~\cite{dynamicalqentanglement}, we will measure the degree of ergodicity of $\proj_I$ over a range of time as weighted by $w(t)$ by the size of the ensemble as given by the purity~\cite{NielsenChuang}, $\mathcal{P}[\hat{\rho}] = \Tr[\hat{\rho}^2]$, of this density operator:
\begin{equation}
    \mathcal{P}_w(\proj_I) \equiv \mathcal{P}[\hat{\rho}_w(\proj_I)] =  \Tr[\hat{\rho}_w(\proj_I)^2].
    \label{eq:puritydef}
\end{equation}
This purity translates to a notion of ``size'' of an ensemble~\cite{StechelMeasure, StechelHeller, dynamicalqergodicity, dynamicalqthermalization}, $\mu(\hat{\rho}) = 1/(D \Tr[\hat{\rho}^2]) \in [0,1]$, representing the nominal fraction of the Hilbert space spanned by a given ensemble; here, it is given by:
\begin{equation}
    \mu_w(\proj_I) \equiv \mu(\hat{\rho}_w(\proj_I)) = \frac{1}{D \mathcal{P}_w(\proj_I)}.
    \label{eq:ensemblesizedef}
\end{equation}
Note that at this stage, calling $\mu_w(\proj_I)$ the nominal fraction is just terminology, inspired by semiclassical intuition~\cite{StechelMeasure, StechelHeller}. We will convert it to a more operational notion of ergodic exploration in Sec.~\ref{sec:ergodiceffdim}.

Using Eq.~\eqref{eq:KrylovEnsembledef} in Eq.~\eqref{eq:puritydef}, we get
\begin{align}
    \mathcal{P}_w(\proj_I) &= \frac{1}{D_I^2}\int\diff t_1\int\diff t_2\ w(t_1)w(t_2) \Tr[\proj_I(t_1)\proj_I(t_2)] \nonumber \\
    &= \frac{1}{D_I^2}\int \diff t\ \left\lbrace\int\diff t_I\ w(t+t_I) w(t_I)\right\rbrace \Tr[\proj_I(t)\proj_I(0)].
\end{align}
Identifying the return probability of the initial subspace,
\begin{equation}
    p_I(t) \equiv \frac{1}{D_I}\Tr[\proj_I(t)\proj_I(0)],
\end{equation}
we have
\begin{equation}
    \mathcal{P}_w(\proj_I) = \frac{1}{D_I} \int\diff t\ p_I(t) \left\lbrace\int\diff t_I\ w(t+t_I) w(t_I)\right\rbrace.
    \label{eq:purity_and_probability_revisited}
\end{equation}
This is a mild generalization of one of the main results of Ref.~\cite{dynamicalqentanglement} concerning the orthonormality of unitary dynamics (with the generalization being to arbitrary weights $w(t)$ from a discrete set of times and to subspaces $\proj_I$ larger than a single pure state), and of an analogous earlier semiclassical result~\cite{HellerPhaseSpaceDynamics} focusing on phase space densities (generalized to arbitrary weights). Its content is that a low return probability $p_I(t)$ at most times $t$ translates to a low purity of the density operator. This, in turn, suggests a larger ensemble of states generated by time evolution, and therefore the ``ergodic'' exploration of a larger portion of the Hilbert space. Specifically, using Eq.~\eqref{eq:purity_and_probability_revisited} in Eq.~\eqref{eq:ensemblesizedef}, we get
\begin{lemma}[Aperiodicity determines the nominal fraction of the Hilbert space explored by a subspace]
\label{lem:ergodicfraction}
    For a given subspace of initial states $\mathcal{H}_I$ with $\proj_I$, whose return probability at the time $t$ is $p_I(t)$, the nominal fraction $\mu_w(\proj_I)$ [defined in Eq.~\eqref{eq:ensemblesizedef}] of the full Hilbert space $\mathcal{H}$ explored by the dynamics of this subspace, as weighted over time by $w(t)$, is given by:
    \begin{equation}
        \mu_w(\proj_I) = \frac{D_I}{\displaystyle{D \int\diff t\ p_I(t) \left\lbrace\int\diff t_I\ w(t+t_I) w(t_I)\right\rbrace}}.
        \label{eq:ergodicity_and_probability}
    \end{equation}
\end{lemma}

\subsection{Operational effective dimension of partially ergodic exploration}
\label{sec:ergodiceffdim}

Now, we would like to obtain a guarantee on how much of the Hilbert space $\mathcal{H}$ is explored by the trajectory of $\proj_I$ in a precise sense, given the nominal fraction of exploration $\mu_w(\proj_I)$ (merely defined above in terms of the purity).

To examine this question, consider an arbitrary test observable $\proj_{\obs}$ given by a projector onto a subspace $\mathcal{H}_{\obs}$ of dimension $D_{\obs} = \Tr[\proj_{\obs}]$. This observable probes the extent to which the history of $\proj_I$ overlaps with $\mathcal{H}_{\obs}$. We can write, using $\proj_{\obs}^2 = \proj_{\obs}$, and the Cauchy-Schwarz inequality~\cite{ByronFuller}:
\begin{equation}
    \Tr[\proj_{\obs} \hat{\rho}_w(\proj_I)] \leq \sqrt{D_{\obs} \Tr\left[\left\lbrace \hat{\rho}_w(\proj_I)\right\rbrace^2\right]} = \sqrt{D_{\obs} \mathcal{P}_w(\proj_I)}. \label{eq:probeprojinequality}
\end{equation}
This is formally similar to Eq.~\eqref{eq:Tasaki_effdim1} in Sec.~\ref{sec:intro} that underlies state-dependent infinite-time considerations~\cite{TasakiTypicalityThermalization}, but differs in a crucial way that will allow us to access finite time results. Where Eq.~\eqref{eq:Tasaki_effdim1} already performs the time average and eliminates dynamics via Eq.~\eqref{eq:Tasaki_infytimeaverage} before applying the Cauchy-Schwarz inequality, the dynamics remains fully implicit in $\hat{\rho}_w(\proj_I)$ in Eq.~\eqref{eq:probeprojinequality} here and is yet to be accessed in our approach.

Notably, Eq.~\eqref{eq:probeprojinequality} remains a statement about the \textit{invariant} (basis-independent) $w(t)$-weighted history of the subspace $\mathcal{H}_I$ as witnessed by the projective observable $\proj_{\obs}$. Its intuitive content can be expressed in two ways: (1) if the initial subspace explores a large fraction of the Hilbert space, then it \textit{necessarily} has negligible overlap with any small subspace $\mathcal{H}_{\obs}$, (2) $D_{\obs}$ needs to be a sufficiently large fraction of the size of the Hilbert space explored by the initial subspace for $\mathcal{H}_{\obs}$ to be able to access a notable fraction of its history. Therefore, in comparison with Eq.~\eqref{eq:Tasaki_effdim1} and \eqref{eq:intro_effdim_def} in Sec.~\ref{sec:intro}, $\mathcal{P}_w(\proj_I)$ plays an operational role as the inverse effective dimension of the invariant \textit{finite-time} history of the subspace (as opposed to just in a basis) if $w(t)$ has support in a finite time interval. This can be formalized by re-expressing Eq.~\eqref{eq:probeprojinequality} in terms of the nominal exploration fraction $\mu_w(\proj_I)$ as follows:
\begin{lemma}[Probing the effective dimension of ergodic exploration]
\label{lem:ergodicprobesubspace}
    Consider the density operator $\hat{\rho}_w(\proj_I)$ representing the history of $\proj_I$ weighted by $w(t)$, exploring a nominal fraction $\mu_w(\proj_I)$ of the Hilbert space $\mathcal{H}$ as in Lemma~\ref{lem:ergodicfraction}.
    \begin{enumerate}
        \item The overlap of this history with any subspace $\mathcal{H}_{\obs} \subseteq \mathcal{H}$ with projector $\proj_{\obs}$ of dimension $D_{\obs} = \dim \mathcal{H}_{\obs} = \Tr[\proj_{\obs}]$ is constrained by:
    \begin{equation}
        \Tr[\proj_{\obs} \hat{\rho}_w(\proj_I)] \leq \sqrt{\frac{D_{\obs}}{D \mu_w(\proj_I)}}.
        \label{eq:smallergodicprobeconstraint}
    \end{equation}
    \item Consequently, if $\proj_{M}$ projects onto a subspace $\mathcal{H}_M$ of dimension $D_M$ containing at least a $(1-\epsilon)$ probability of the $w(t)$-weighted history of $\proj_I$ for $\epsilon \in [0,1]$,
    \begin{equation}
        \Tr[\proj_M \hat{\rho}_w(\proj_I)] \geq (1-\epsilon),
    \end{equation}
    then (by assigning $\mathcal{H}_{\obs} := \mathcal{H}_M$ above) this subspace must have a dimension of at least
    \begin{equation}
        D_M \geq (1-\epsilon)^2 D \mu_w(\proj_I),
        \label{eq:largeergodicprobeconstraint}
    \end{equation}
    where $D_M \equiv \dim \mathcal{H}_M = \Tr[\proj_M]$.
    \end{enumerate}
    Thus, $\mu_w(\proj_I)$ directly constrains the fraction of the Hilbert space explored by the trajectory of $\proj_I$ as visible to probes based on projective measurements.
\end{lemma}

The second part of Lemma~\ref{lem:ergodicprobesubspace} in combination with Lemma~\ref{lem:ergodicfraction} gives the first part of Theorem~\ref{thm:aperiodicityimpliesmacroscopicthermalization}.

\subsection{Highly ergodic histories are mostly in macroscopic equilibrium}

Now, let us focus on the nonequilibrium subspace (see Sec.~\ref{sec:concentratedobservables} for context) $\mathcal{H}_{\neqm}$ with projector $\proj_{\neqm}$ of dimension $D_{\neqm} = \dim \mathcal{H}_{\neqm} = \Tr[\proj_{\neqm}]$, which is usually (but not necessarily) small compared to $D$. Its complement is the equilibrium subspace $\mathcal{H}_{\eqm}$ with projector $\proj_{\eqm}$ of dimension $D_{\eqm} = \mathcal{H}_{\eqm} = \Tr[\proj_{\eqm}]$, satisfying $\mathcal{H} = \mathcal{H}_{\eqm} \oplus \mathcal{H}_{\neqm}$, and correspondingly $\idop = \proj_{\eqm} + \proj_{\neqm}$ with $D = D_{\eqm} + D_{\neqm}$.

Then, using the complementary relationship between the equilibrium and nonequilibrium subspaces and connecting to partially ergodic dynamics via Eq.~\eqref{eq:smallergodicprobeconstraint} in the first part of Lemma~\ref{lem:ergodicprobesubspace} (specifically, assigning $\proj_{\obs} := \proj_{\neqm}$), we have
\begin{equation}
    \Tr[\proj_{\eqm} \hat{\rho}_w(\proj_I)]  = 1-\Tr[\proj_{\neqm} \hat{\rho}_w(\proj_I)] \geq 1-\sqrt{\frac{D_{\neqm}}{D\mu_w(\proj_I)}}.
    \label{eq:equilibriumsubspaceergodicityconstraint}
\end{equation}

Now, we can decompose $\proj_I$ as an equal ensemble of $D_I$ orthonormal pure states $\lbrace \lvert \psi_{Ik}\rangle\rbrace_{k=1}^{D_I}$,
\begin{equation}
    \proj_I = \sum_{k=1}^{D_I} \lvert \psi_{Ik}\rangle \langle \psi_{Ik}\rvert
\end{equation}
the choice of the latter being non-unique (highly degenerate). Then, from Eq.~\eqref{eq:KrylovEnsembledef}, $\hat{\rho}_w(\proj_I)$ can be written as an average over the histories of each such pure state:
\begin{equation}
    \hat{\rho}_w(\proj_I) = \int\diff t\ w(t)\frac{1}{D_I}\sum_{k=1}^{D_I} \lvert \psi_{Ik}(t)\rangle \langle \psi_{Ik}(t)\rvert,
\end{equation}
where $\lvert \psi_{Ik}(t)\rangle = e^{-i\hat{H} t}\lvert \psi_{Ik}\rangle$.

For the pure states themselves, Eq.~\eqref{eq:equilibriumsubspaceergodicityconstraint} then implies
\begin{equation}
    \frac{1}{D_I}\sum_{k=1}^{D_I} \int\diff t\ w(t)\langle \psi_{Ik}(t)\rvert \proj_{\neqm}\lvert \psi_{Ik}(t)\rangle \leq \sqrt{\frac{D_{\neqm}}{D\mu_w(\proj_I)}}.
    \label{eq:nonequilibriumexplorationfraction}
\end{equation}
The left hand side is a weighted average over nonnegative quantities $\langle \psi_{Ik}(t)\rvert \proj_{\eqm}\lvert \psi_{Ik}(t)\rangle \geq 0$, allowing the use of Markov's inequality~\cite{RossProbability} to show that most of the trajectories have a small overlap with the nonequilibrium subspace if a large fraction of the Hilbert space is explored. Specifically, given some (small) resolution $\lambda > 0$, define the $D_I$ different temporal domains
\begin{equation}
    \mathcal{T}_k = \left\lbrace t:\ \langle \psi_{Ik}(t)\rvert \proj_{\neqm}\lvert \psi_{Ik}(t)\rangle \geq \lambda \right\rbrace\;\; \text{ for } k \in \mathbb{N}_{D_I},
\end{equation}
where $\mathbb{N}_n = \lbrace 1,\ldots,n\rbrace$ denotes the first $n$ natural numbers.
Then, it follows from Eq.~\eqref{eq:nonequilibriumexplorationfraction} that
\begin{equation}
    \frac{1}{D_I}\sum_{k=1}^{D_I} \int_{\mathcal{T}_k}\diff t\ w(t)\ \leq\ \frac{1}{\lambda} \sqrt{\frac{D_{\neqm}}{D\mu_w(\proj_I)}}.
    \label{eq:weightednonequilibriumtimefraction}
\end{equation}
Consistent with intuition, we note that a larger value of $\mu_w(\proj_I)$ (corresponding to exploring a larger fraction of the Hilbert space) suppresses the (weighted) amount of time spent with a non-negligible overlap of at least $\lambda$ with the nonequilibrium subspace.

For convenience, define the average weighted time spent by the $\lvert \psi_{Ik}\rangle$ with an overlap of at least $(1-\lambda)$ in the equilibrium subspace:
\begin{equation}
    T_w\left[\left\lbrace \lvert \psi_{Ik}\rangle\right\rbrace_{k=1}^{D_I} \right](\lambda) \equiv \frac{1}{D_I}\sum_{k=1}^{D_I} \int \diff t\ w(t)\ \ifc\left(\langle \psi_{Ik}(t)\rvert \proj_{\eqm}\lvert \psi_{Ik}(t)\rangle \geq 1-\lambda\right),
        \label{eq:equilibriumtimefraction_def}
\end{equation}
Then we have from Eq.~\eqref{eq:weightednonequilibriumtimefraction}:
\begin{lemma}[High ergodicity implies macroscopic equilibrium for almost all states and times]
\label{lem:ergodicityimpliesequilibrationfinitetime}
    Let $\lbrace \lvert \psi_{Ik}\rangle\rbrace_{k=1}^{D_I} \subset \mathcal{H}$ be a set of $D_I \leq D$ orthonormal pure states. Form the projector $\proj_I = \sum_{k=1}^{D_I} \lvert \psi_{Ik}\rangle \langle \psi_{Ik}\rvert$, and let its $w(t)$-weighted history explore a nominal fraction $\mu_w(\proj_I)$ of the Hilbert space $\mathcal{H}$. Given the equilibrium subspace $\mathcal{H}_{\eqm} \subseteq \mathcal{H}$ of dimension $D_{\eqm} = \dim \mathcal{H}_{\eqm}$, for any $\lambda \in (0,1]$, the weighted fraction of times $T_w\left[\left\lbrace \lvert \psi_{Ik}\rangle\right\rbrace_{k=1}^{D_I} \right](\lambda)$ spent on average by the states $\lvert \psi_{Ik}\rangle$ with an overlap of at least $1-\lambda$ with $\mathcal{H}_{\eqm}$, as defined in Eq.~\eqref{eq:equilibriumtimefraction_def}, is constrained by:
    \begin{equation}
        T_w\left[\left\lbrace \lvert \psi_{Ik}\rangle\right\rbrace_{k=1}^{D_I} \right](\lambda)\ \geq\ \left[1-\frac{1}{\lambda} \sqrt{\frac{D - D_{\eqm}}{D\mu_w(\proj_I)}}\right].
    \end{equation}
\end{lemma}

\subsection{Finite-time ergodicity to long-time thermalization}

Lemma~\ref{lem:ergodicityimpliesequilibrationfinitetime} shows that thermalization occurs for a large $w$-weighted fraction of times provided one has highly ergodic dynamics over the same $w$-weighted range of times. In practice, if $w$ has compact support, this constrains thermalization in the same compact set of times in which ergodicity can be established. As it turns out, it is straightforward to extrapolate from this compact set to all times.

Intuitively, this is because of the time-translation invariance and unitarity of Hamiltonian dynamics~\cite{DiracQM, ShankarQM}: time evolution across $[t_1,t_1+t]$ is isomorphic to time evolution across $[0,t]$ up to a unitary transformation $e^{-i\hat{H}t_1}$. In particular, if $f(\lvert \phi\rangle, \lvert \chi\rangle, \ldots)$ is a multi-argument function on $\mathcal{H}$ invariant under simultaneous unitary transformations of its arguments:
\begin{equation}
   f(\hat{U}\lvert \phi\rangle, \hat{U}\lvert \chi\rangle, \ldots) = f(\lvert \phi\rangle, \lvert \chi\rangle, \ldots)\;\; \text{ for all }\ \hat{U}: \hat{U}^{-1} = \hat{U}^\dagger,
\end{equation}
then
\begin{equation}
    f(e^{-i\hat{H} (t+t_1)}\lvert \phi\rangle, e^{-i\hat{H} t_1}\lvert \chi\rangle, \ldots) = f(e^{-i\hat{H} t}\lvert \phi\rangle, \lvert \chi\rangle, \ldots).
\end{equation}
Therefore, as our measures of ergodicity (from Sec.~\ref{sec:aperiodicityandergodicity}) such as $\mathcal{P}_w(\proj_I)$ and $\mu_w(\proj_I)$ are unitarily invariant functions of $\proj_I(t)$ (which can be viewed as joint functions of multiple time-evolved states in the subspace $\mathcal{H}_I$), they are invariant under time translations. Specifically, a certain degree of ergodicity in $[0,t]$ implies the \textit{same} degree of ergodicity in $[t_1, t_1+t]$.

To see this more formally, let $S_R \equiv \left(t_r\right)_{r=1}^{R}$ be a discrete sequence of time shifts with $R \in \mathbb{N}$ and $t_{r+1} > t_r$. Given the weighting function $w(t)$, define its $r$-th time translated variant:
\begin{equation}
    w_r(t) \equiv w(t+t_r).
\end{equation}
Then, from Eqs.~\eqref{eq:purity_and_probability_revisited} and \eqref{eq:ergodicity_and_probability}, on account of
\begin{equation}
    \int\diff t_I\ w_r(t+t_I) w_r(t_I) = \int\diff t_I\ w(t+t_I+t_r) w(t_I+t_r) = \int\diff t_I\ w(t+t_I) w(t_I),
\end{equation}
we have
\begin{equation}
    \mathcal{P}_{w_r}(\proj_I) = \mathcal{P}_w(\proj_I);\ \mu_{w_r}(\proj_I) = \mu_w(\proj_I)
\end{equation}
for any such sequence of time shifts, for any $w(t)$ and $\proj_I$.

More generally, we are interested in patching together different weights to form a larger interval of time. Consider the averaged weight function:
\begin{equation}
    w_{S_R}(t) \equiv \frac{1}{R}\sum_{r=1}^{R} w_r(t).
    \label{eq:averageshiftedweightdef}
\end{equation}
For example, let $w(t) = \ifc(t\in[0,T])/T$, and $t_r = (r-1)T$. Then $w_{S_R}(t) = \ifc(t\in [0,RT])/(RT)$, which expands the time interval of interest from $[0,T]$ to $[0,RT]$ for any $R \in \mathbb{N}$; note that we can also take $R \to \infty$ to extrapolate to $[0,\infty)$ from $[0,T]$. Then, the $w_{S_R}(t)$ weighted fraction of times spent in the equilibrium subspace, defined in Eq.~\eqref{eq:equilibriumtimefraction_def}, decomposes into a simple average:
\begin{equation}
    T_{w_{S_R}}\left[\left\lbrace \lvert \psi_{Ik}\rangle\right\rbrace_{k=1}^{D_I} \right](\lambda) = \frac{1}{R}\sum_{k=1}^{R}T_{w_r}\left[\left\lbrace \lvert\psi_{Ik}\rangle\right\rbrace_{k=1}^{D_I} \right](\lambda).
\end{equation}
Applying Lemma~\ref{lem:ergodicityimpliesequilibrationfinitetime} to each term on the right hand side independently gives:
\begin{lemma}[Finite time ergodicity implies long time thermalization]
\label{lem:finitetimeergodicitytoinfinitetimethermalization}
    For the same setting as Lemma~\ref{lem:ergodicityimpliesequilibrationfinitetime} based on the weight function $w(t)$, for any sequence of time shifts $S_R$ and an average shifted weight $w_{S_R}(t)$ defined in Eq.~\eqref{eq:averageshiftedweightdef}, we have the bound:
    \begin{equation}
        T_{w_{S_R}}\left[\left\lbrace \lvert \psi_{Ik}\rangle\right\rbrace_{k=1}^{D_I} \right](\lambda)\ \geq\ \left[1-\frac{1}{\lambda} \sqrt{\frac{D - D_{\eqm}}{D\mu_w(\proj_I)}}\right].
    \end{equation}
\end{lemma}
Combining Lemma~\ref{lem:finitetimeergodicitytoinfinitetimethermalization} (which derives from Lemma~\ref{lem:ergodicityimpliesequilibrationfinitetime}, which in turn derives from the first part of Lemma~\ref{lem:ergodicprobesubspace}) with Lemma~\ref{lem:ergodicfraction} for $w(t) = \ifc(t\in [0,T])$ and $S_R = \left((r-1)T\right)_{r=1}^{R}$ gives the second statement of Theorem~\ref{thm:aperiodicityimpliesmacroscopicthermalization}.

\section{Corollaries}
\label{sec:corollaries}

\subsection{Initial state delocalization over energy shells}
\label{sec:corollaries_delocalization}

Now, we will derive a connection between Theorem~\ref{thm:aperiodicityimpliesmacroscopicthermalization} and a collection of results linking equilibration to the delocalization of initial states in the energy eigenbasis~\cite{Tasaki1998, ReimannRealistic, LindenetalEqb, ShortEqb, ShortDegenerate, TasakiTypicalityThermalization}. This connection is a corollary of (in some parts, a very mild generalization of) Lemma~\ref{lem:ergodicfraction}.

The intuition behind this connection is straightforward in its eigenstate-specific form (which we will later generalize). As energy eigenstates are stationary under time evolution, no superposition of a subset of energy eigenstates can exit the subspace spanned by the subset. Therefore, the initial state must be delocalized over at least as many energy eigenstates as the number of dimensions of the Hilbert space that it has been verified to ergodically explore. The main technical task here is to sharpen this statement to refer to the specific ``effective dimension''\footnote{In Ref.~\cite{LindenetalEqb}, the effective dimension is directly defined in terms of the purity, but applied to the infinite-time averaged density operator, which amounts to Eq.~\eqref{eq:effdimdef} that subsequent works~\cite{ShortEqb, ShortDegenerate, TasakiTypicalityThermalization} take as the direct definition. This allows interpreting it~\cite{LindenetalEqb} as the number of states explored over infinite time, but such an infinite time exploration cannot be computationally verified. Our goal here is to connect it instead to finite-time partially ergodic exploration, which can be observed operationally~\cite{dynamicalqentanglement} through entanglement measurements or quantum teleportation protocols.} measure of energy delocalization~\cite{LindenetalEqb, ShortEqb, ShortDegenerate, TasakiTypicalityThermalization} of a density operator $\hat{\rho}$:
\begin{equation}
    d_{\effdim}[\hat{\rho}] \equiv \frac{1}{\sum_n \langle E_n\rvert \hat{\rho}\lvert E_n\rangle^2},
    \label{eq:effdimdef}
\end{equation}
loosely meaning that $\hat{\rho}$ is significantly distributed over $d_{\effdim}[\hat{\rho}]$ energy levels.

To carry out this task, let us write Eq.~\eqref{eq:ergodicity_and_probability} in Lemma~\ref{lem:ergodicfraction}, which continues to apply for any initial density operator $\hat{\rho}_I$ in place of $\proj_I/D_I$, in terms of the energy eigenbasis:
\begin{equation}
    D \mu_w(D_I \hat{\rho}_I) = \frac{1}{\displaystyle{\sum_{n,m} \left\lvert \langle E_n\rvert \hat{\rho}_I\lvert E_m\rangle \right\rvert^2 \widetilde{w}(E_n-E_m) \widetilde{w}(E_m-E_n) }}.
\end{equation}
Here, $\widetilde{w}(t) \equiv \int\diff t \ w(t) e^{-iEt}$. Recall that $\widetilde{w}(0) =1$ [Eq.~\eqref{eq:w_normalization}] and $w(t) \geq 0$ by assumption. If $w(t) \in \mathbb{R}$ (containing this assumption), then $\widetilde{w}(-E) = \widetilde{w}^\ast(E)$, giving us a denominator that is a sum of non-negative terms:
\begin{equation}
    D \mu_w(D_I \hat{\rho}_I) = \frac{1}{\displaystyle{\sum_{n,m} \left\lvert \langle E_n\rvert \hat{\rho}_I\lvert E_m\rangle \right\rvert^2 \left\lvert \widetilde{w}(E_n-E_m) \right\rvert^2 }}.
    \label{eq:energybandinitialstate}
\end{equation}
Now, similar to the connection between finite-time autocorrelators and (weak) eigenstate thermalization in Ref.~\cite{dynamicalqthermalization}, dropping the $n \neq m$ terms from the denominator (as initially done in semiclassical results~\cite{Zelditch}) gives the inequality:
\begin{equation}
    D \mu_w(D_I \hat{\rho}_I) \leq \frac{1}{\sum_n \langle E_n\rvert \hat{\rho}_I\lvert E_n\rangle^2} = d_{\effdim}[\hat{\rho}_I].
    \label{eq:effdim_fraction_constraint1}
\end{equation}
Therefore, e.g. taking $w(t) = \ifc(t\in [0,T])/T$, we get from the definition of aperiodicity [Definition~\ref{def:aperiodicity}] for an initial density operator $\hat{\rho}_I = \proj_I/D_I$ corresponding to the subspace $\mathcal{H}_I$:
\begin{equation}
    d_{\effdim}[\hat{\rho}_I] \geq D_I \mathcal{A}_{[0,T]}[\mathcal{H}_I],
    \label{eq:effdim_fraction_constraint2}
\end{equation}
technically formalizing the intuition that a finite time partially ergodic exploration of the Hilbert space is sufficient to guarantee a corresponding (at least partial) delocalization of the initial state in the energy eigenbasis.

It is worth briefly commenting on the connection between these considerations and equilibration dynamics. Having $d_{\effdim}[\hat{\rho}_I] \gg 1$ (even if e.g. $O(1)$) is sufficient to guarantee the equilibration of all (including local) observables over infinite times (and exponentially long times that diverge to infinity in the thermodynamic limit) to a resolution $\varepsilon \ll 1$ in generic systems~\cite{ShortEqb, ShortDegenerate}, due to the guaranteed destructive interference between a large number of phases $e^{-i(E_n-E_m)t}$ in the energy eigenbasis over such long timescales. Such general guarantees of equilibration also do not exist for finite timescales~\cite{ShortDegenerate}, except possibly with assumptions on the joint matrix element structure of specific observables and the initial state in the energy eigenbasis typically for highly mixed (rather than pure) states~\cite{GarciaPintosFiniteTimeEqb}, which do not appear to be operationally verifiable in general (but may be discussed heuristically in general \cite{dephasing1}). For example, in typical systems (e.g. with random matrix spectral statistics), projectors onto a basis of discrete Fourier transforms of energy eigenstates can only equilibrate after exponentially long times~\cite{dynamicalqergodicity} (divergent in the thermodynamic limit) for states of the same basis, ruling out observable-independent finite-time equilibration results for highly delocalized states. However, observable-dependent (and state-independent) operational criteria~\cite{dynamicalqthermalization, dynamicalpurestatethermalization} can establish finite-time equilibration for almost all initial states with at worst rare exceptions (amounting to at most a finite but small fraction of a basis). In our view, this makes observable-dependent criteria most directly relevant for finite-time microscopic statistical mechanics, with initial-state-dependent criteria as developed in this work primarily being necessitated for macroscopic statistical mechanics, when the manifold of interesting states is sufficiently small to be beyond the reach of finite-resolution guarantees stemming from observable-dependent criteria.

On top of connecting to the effective dimension quantity used in infinite-time state-dependent equilibration results, our main corollary of interest for delocalization in the energy-eigenbasis concerns a stronger statement regarding the eigenstate structure of a \textit{pure} initial state $\lvert \psi_I\rangle$, over \textit{finite} energy scales $\Delta E \sim 1/T$, that follows from aperiodicity over the finite time scale $T$. This is analogous to finite-energy-resolution statements on energy-band thermalization, stronger than (weak) eigenstate thermalization, that can be obtained from the decay of the observable's autocorrelator~\cite{dynamicalqthermalization}. In particular, our result is that if $\lvert \psi_I\rangle$ (corresponding to $\mathcal{H}_I = \lbrace \lvert \psi_I\rangle\rbrace$ with $D_I = 1$) explores $d_{\effdim} \sim \mathcal{A}_{[0,T]}[\lbrace \lvert \psi_I\rangle\rbrace]$ dimensions of the Hilbert space over a time $T$, then it must be spread out over at least $d_{\effdim}$ \textit{different} energy shells of width $\Delta E \sim 1/T$, rather than just energy eigenstates (which may be consecutive in principle, spanning an exponentially small energy window $\sim d_{\effdim} e^{-N}$ for finite $d_{\effdim}$, and therefore does not constrain the macroscopic delocalization structure of the initial state without our stronger claim).

To formalize this, we return to Eq.~\eqref{eq:energybandinitialstate}, and consider an energy width $\Delta E \geq 0$ such that $|\widetilde{w}(E :|E| \leq \Delta E)|^2 \geq W$, for some $W \in (0,1]$. Then we obtain
\begin{equation}
    D \mu_w(\hat{\rho}_I) \leq \frac{1}{\displaystyle{W \sum_{\substack{n,m:\\|E_n-E_m|\leq \Delta E}} \left\lvert \langle E_n\rvert \hat{\rho}_I\lvert E_m\rangle \right\rvert^2  }}.
    \label{eq:delocalizationshell1}
\end{equation}
Now, partition the energy spectrum into $n_E$ disjoint ``energy shells'' $\lbrace \mathcal{E}_r\rbrace_{r=1}^{n_E}$ (such that $\bigcup_{r=1}^{n_E} \mathcal{E}_r = \lbrace E_n\rbrace_{n=1}^D$, and $\mathcal{E}_r \cap \mathcal{E}_s = \lbrace \rbrace$ for $r \neq s$) of respective energy widths $\lbrace \delta E_r\rbrace_{r=1}^{n_E}$, such that $\delta E_r \leq \Delta E$ (i.e. $\lvert E_n-E_m\rvert \leq \delta E_r$ for all $E_n, E_m \in \mathcal{E}_r$ with at least one pair of $(n,m)$ saturating the inequality), and with projectors $\lbrace \proj(\mathcal{E}_r)\rbrace_{r=1}^{n_E}$. Then, we have
\begin{equation}
    \sum_{\substack{n,m:\\|E_n-E_m|\leq \Delta E}} \left\lvert \langle E_n\rvert \hat{\rho}_I\lvert E_m\rangle \right\rvert^2 \geq \sum_{r=1}^{n_E} \sum_{\substack{m,n:\\ E_m,E_n \in \mathcal{E}_r}} \left\lvert \langle E_n\rvert \hat{\rho}_I\lvert E_m\rangle \right\rvert^2.
    \label{eq:delocalizationshell2}
\end{equation}
At this stage, it is convenient to explicitly take our state to be a pure state, $\hat{\rho}_I = \lvert \psi_I\rangle \langle \psi_I\rvert$. Then, we have (amounting to $|\rho_{nm}|^2 = \rho_{nn}\rho_{mm}$ for the matrix elements $\rho_{nm}$ of a pure state $\hat{\rho} = \hat{\rho}^2$ in any orthonormal basis):
\begin{align}
    \sum_{r=1}^{n_E} \sum_{\substack{m,n:\\ E_m,E_n \in \mathcal{E}_r}} \left\lvert \langle E_n\rvert \hat{\rho}_I\lvert E_m\rangle \right\rvert^2 = \sum_{r=1}^{n_E} \sum_{\substack{m,n:\\ E_m,E_n \in \mathcal{E}_r}} \langle E_n\vert \psi_I\rangle \langle \psi_I\vert E_m\rangle\langle E_m\vert \psi_I\rangle \langle \psi_I\vert E_n\rangle &= \sum_{r=1}^{n_E} \left(\sum_{\substack{n: E_n \in \mathcal{E}_r}} \langle \psi_I\vert E_n\rangle\langle E_n\vert \psi_I\rangle\right)^2 \nonumber \\
    &=\sum_{r=1}^{n_E} \left(\langle \psi_I\rvert \proj(\mathcal{E}_r)\lvert \psi_I\rangle\right)^2.
    \label{eq:delocalizationshell3}
\end{align}
Let $P(\lvert \psi_I\rangle, \mathcal{E}_r) = \langle \psi_I\rvert \proj(\mathcal{E}_r)\lvert \psi_I\rangle$, the overlap of the state with the $r$-th energy shell, and note that $\sum_{r=1}^{n_E} P(\lvert \psi_I\rangle, \mathcal{E}_r)  = 1$. Define the ``effective participation dimension'' of $\lvert \psi_I\rangle$ with respect to the partition $\lbrace \mathcal{E}_r\rbrace_{r=1}^{n_E}$ as:
\begin{equation}
    d_{\effdim}\left(\lvert \psi_I\rangle, \lbrace \mathcal{E}_r\rbrace_{r=1}^{n_E}\right) \equiv \frac{1}{\displaystyle{\sum_{r=1}^{n_E} P^2(\lvert \psi_I\rangle, \mathcal{E}_r)}}.
    \label{eq:shellparticipationratio_def}
\end{equation}
Then, on using Eqs.~\eqref{eq:delocalizationshell2} and \eqref{eq:delocalizationshell3} in Eq.~\eqref{eq:delocalizationshell1}, we get
\begin{equation}
    d_{\effdim}\left(\lvert \psi_I\rangle, \lbrace \mathcal{E}_r\rbrace_{r=1}^{n_E}\right) \geq W D \mu_w(\lvert \psi_I\rangle \langle \psi_I\rvert).
    \label{eq:energyshelldelocalization1}
\end{equation}
Loosely, this means that $\lvert \psi_I\rangle$, if it explores a fraction $\mu_w(\lvert \psi_I\rangle\langle \psi_I\rvert)$ of the Hilbert space under $w(t)$-weighted dynamics, must be delocalized over at least $W D \mu_w(\lvert \psi_I\rangle\langle \psi_I\rvert)$ shells in any partition of the spectrum into distinct energy shells of width at most $\Delta E$.


It is worth emphasizing a contrast between states and observables: where the analogous ``energy-band thermalization'' property of observables with highly degenerate eigenspaces, inferred from their decaying autocorrelators, implies that the observable uniformly approaches its thermal value in each energy shell~\cite{dynamicalqthermalization} with only small deviations (as in eigenstate thermalization), much larger fluctuations between energy shells are allowed for pure states by the participation constraint in Eq.~\eqref{eq:energyshelldelocalization1} while $d_{\effdim} < n_E$ (e.g. some energy shells may have significant participation and others may have very little). For random states, this nonuniformity may often not be visible for large $\Delta E$ (small $n_E$)~\cite{vonNeumannThermalization}, but becomes increasingly relevant as $\Delta E \to 0$ ($n_E \to D$). Also note that Eq.~\eqref{eq:energyshelldelocalization1} is a considerably stronger statement than the constraint on effective dimension in Eq.~\eqref{eq:effdim_fraction_constraint1}, which is only implied as a special case when $\Delta E = 0$ (with $W = 1$). For example, Eq.~\eqref{eq:energyshelldelocalization1} gets considerably closer than Eq.~\eqref{eq:effdim_fraction_constraint1} towards providing an operational criterion showing that a given initial state of interest is delocalized over macroscopic energy scales, for which numerical evidence (consistent with energy-dependent random delocalization) is seen in small systems of interest in Ref.~\cite{FoiniDymarskyPappalardiInitialStateETH}.

To express this in terms of the observable aperiodicity of the initial state as in Definition~\ref{def:aperiodicity}, we again take $w(t) = \ifc(t\in[0,T])/T$, corresponding to $|\widetilde{w}(E)| = \sinc(ET/2)$, and note that $W = \sinc^2(\Delta E T/2)$ is a valid choice for any $\Delta E < 2\pi/T$. Then, we get
\begin{corollary}[Aperiodicity or partial ergodicity implies energy shell delocalization for a pure state]
\label{cor:energyshelldelocalization}
For any given $\Delta E > 0$, consider a partition of the energy levels $\lbrace E_n\rbrace_{n=1}^{D}$ of the Hamiltonian $\hat{H}$ in the Hilbert space $\mathcal{H}$ into $n_E$ disjoint energy shells $\lbrace \mathcal{E}_r\rbrace_{r=1}^{n_E}$ with respective projectors $\proj(\mathcal{E}_r)$ (satisfying $\proj(\mathcal{E}_r)\proj(\mathcal{E}_s) = \proj(\mathcal{E}_r)\delta_{rs}$ and $\sum_r \proj(\mathcal{E}_r) = \idop$) such that each shell has a width of at most $\Delta E$:
\begin{equation}
    \text{for each $r$, }\ |E_n-E_m| \leq \Delta E \text{ for all } E_n, E_m \in \mathcal{E}_r.
    \label{eq:partitionconstraint}
\end{equation}
Then, for any normalized initial state $\lvert \psi_I\rangle$ satisfying $\langle \psi_I\vert \psi_I\rangle = 1$, its nominal fraction $\mu_{\ifc(t\in[0,T])/T}(\lvert \psi_I\rangle\langle\psi\rvert)$ of Hilbert space explored (Lemma~\ref{lem:ergodicfraction}) or equivalently its aperiodicity $\mathcal{A}_{[0,T]}[\lbrace \lvert \psi_I\rangle\rbrace]$ (Definition~\ref{def:aperiodicity}) in the time interval $t \in [0,T]$ constrains its participation dimension $d_{\effdim}\left(\lvert \psi_I\rangle, \lbrace \mathcal{E}_r\rbrace_{r=1}^{n_E}\right)$ among such energy shells [defined in Eq.~\eqref{eq:shellparticipationratio_def}] according to:
\begin{align}
    d_{\effdim}\left(\lvert \psi_I\rangle, \lbrace \mathcal{E}_r\rbrace_{r=1}^{n_E}\right) &\geq \frac{4 \sin^2(\Delta E T/2)}{\Delta E^2 T^2} D\mu_{\ifc(t\in[0,T])/T}(\lvert \psi_I\rangle\langle\psi\rvert) \nonumber \\
    &= \frac{4 \sin^2(\Delta E T/2)}{\Delta E^2 T^2}\mathcal{A}_{[0,T]}[\lbrace \lvert \psi_I\rangle\rbrace],
    \label{eq:energyshelldelocalization2}
\end{align}
for every partition satisfying Eq.~\eqref{eq:partitionconstraint} for every $\Delta E < 2\pi/T$. For $\Delta E \to 0$, this reduces to constraints on the conventional effective dimension $d_{\effdim}[\lvert \psi_I\rangle\langle\psi_I\rvert]$ of the pure initial state, defined in Eq.~\eqref{eq:effdimdef}, as per Eqs.~\eqref{eq:effdim_fraction_constraint1} and \eqref{eq:effdim_fraction_constraint2}.
\end{corollary}


\subsection{The thermal fate of a pure state}
\label{sec:corollaries_purestate}

Ideally, we would like to use aperiodicity to determine if all macroscopic observables thermalize within some finite time in a given pure initial state. But such a determination cannot be fully made in general: the aperiodicity of a pure state in a $D$-dimensional Hilbert space only certifies that the state largely explores some subset of the $(D-1)$ other (orthonormal) states in the Hilbert space. Over finite times, this subset may entirely lie within the nonequilibrium subspace of a specific observable, and therefore prevent the certification of thermalization of all macroscopic observables.

To see this quantitatively, we can use the Mandelstam-Tamm (MT)~\cite{MT, AAMT, QSLreview1, GongHamazakiNoneqbBounds} formulation of the energy-time uncertainty principle. Let us focus on a single initial state $\lvert \psi_I\rangle$, with energy standard deviation:
\begin{equation}
    \sigma_{EI} = \sqrt{\langle \psi_I\rvert \hat{H}^2\lvert \psi_I\rangle -  (\langle \psi_I\rvert \hat{H}\lvert \psi_I\rangle)^2}.
\end{equation}
Then, the MT energy-time uncertainty principle sets a speed limit on the return probability:
\begin{equation}
    p_I(t,t') = \left\lvert \langle \psi_I(t')\vert \psi_I(t)\rangle\right\rvert^2 \geq \begin{dcases}
        \cos^2[\sigma_{EI} (t-t')],&\ \text{ for } |t-t'| \leq \frac{\pi}{2\sigma_{EI}}, \\
        0,&\ \text{ otherwise }.
    \end{dcases}
\end{equation}
Consequently, by Eq.~\eqref{eq:aperiodicitydef}, the aperiodicity in a time interval $\mathcal{T} = [0,T]$ satisfies (extending $T = \pi/\sigma_{EI}$ bounds to $T < \pi/\sigma_{EI}$ as well, as the resulting bound on aperiodicity is a monotonically increasing function of $T$):
\begin{equation}
    \mathcal{A}_{[0,T]}[\lbrace \lvert \psi_I\rangle\rbrace] \leq \frac{8 \max\lbrace \pi^2,\sigma_{EI}^2 T^2\rbrace}{4-\pi^2 + 4\pi\max\lbrace \pi, \sigma_{EI} T\rbrace}.
    \label{eq:aperiodicityspeedlimit}
\end{equation}
Asymptotically, we have $\mathcal{A}_{[0,T]}[\lbrace \lvert \psi_I\rangle\rbrace] \lesssim 2\sigma_{EI} T/\pi$, at times $T \gg 1$.

In the time interval $[0,RT]$ for $R > 1$, we would like Theorem~\ref{thm:aperiodicityimpliesmacroscopicthermalization} to imply that the pure state spends at least a fraction $\tau_{\eqm}$ of times with an overlap of at least $(1-\varepsilon)$ with some equilibrium subspace $\proj_{\eqm}$ of dimension $D_{\eqm}$. We recall from the discussion immediately following Definition~\ref{def:aperiodicity} that such a guarantee is necessary for the equilibration of \textit{all} conceivable concentrated observables. Then, from Eq.~\eqref{eq:equilibrationfractionfromaperiodicity} in the statement of Theorem~\ref{thm:aperiodicityimpliesmacroscopicthermalization}, this corresponds to imposing:
\begin{equation}
    \tau_{\eqm} \leq  \left[1-\frac{1}{\varepsilon}\sqrt{\frac{D - D_{\eqm}}{\mathcal{A}_{[0,T]}[\lbrace \lvert \psi_I\rangle\rbrace]}}\right],
\end{equation}
or, after some rearrangement,
\begin{equation}
    \mathcal{A}_{[0,T]}[\lbrace \lvert \psi_I\rangle\rbrace] \geq \frac{D-D_{\eqm}}{\varepsilon^2(1-\tau_{\eqm})^2}.
    \label{eq:aperiodicityrequired}
\end{equation}
From the qualitative behavior of Eqs.~\eqref{eq:aperiodicityspeedlimit} and \eqref{eq:aperiodicityrequired}, we conclude that we need to establish aperiodicity of a pure state over a timescale of at least $T \sim D_{\neqm}/\sigma_{EI}$ to constrain macroscopic thermalization in observables whose nonequilibrium subspace has dimension at most $D_{\neqm}$. Recalling from Eq.~\eqref{eq:densities_Dneqm_bound} that for macroscopic observables such as charge densities, $D_{\neqm} \sim 2^{\Theta(N)}$ is exponentially large in the system size $N$, while $\sigma_{EI}$ is typically at most polynomially large~\cite{BukovSelsPolkovnikov}, we require an exponentially long time $T \sim 2^{\Theta(N)}$ as well. Via Corollary~\ref{cor:energyshelldelocalization}, this corresponds to the ``moderate'' energy distribution ($d_{\effdim} \sim \exp(N) \ll D$) of the initial state required in Ref.~\cite{TasakiTypicalityThermalization} to show equilibration over infinite times. The formal asymptotic bound underlying these conclusions is:

\begin{corollary}[Pure states need macroscopically sustained aperiodicity for Theorem~\ref{thm:aperiodicityimpliesmacroscopicthermalization} to guarantee macroscopic thermalization]
\label{cor:purestatethermalization}
    As $D_{\neqm} \to \infty$, the length $T$ of the time interval $[0,T]$ required for the aperiodicity of an initial state $\lvert \psi_I\rangle$ with energy spread (standard deviation) $\sigma_{EI}$ to guarantee the macroscopic thermalization of all concentrated observables with a nonequilibrium subspace of dimension at most $D_{\neqm}$, via Theorem~\ref{thm:aperiodicityimpliesmacroscopicthermalization}, with resolution $\varepsilon$ for at least a fraction $\tau_{\eqm}$ of times satisfies the asymptotic bound:
    \begin{equation}
        T \gtrsim \frac{\pi D_{\neqm}}{2\sigma_{EI} \epsilon^2(1-\tau_{\eqm})^2}.
    \end{equation}
    In particular, this time scales linearly with $D_{\neqm}$, and is therefore macroscopically long relative to $\sigma_{EI}^{-1}$ if $D_{\neqm}$ is macroscopically large.
\end{corollary}

\subsection{Finite-time thermalization of nonequilibrium states}

To obtain genuine finite-time thermalization results, we need a stronger form of aperiodicity: not just of a pure state, but of a large ensemble of pure states i.e. the subspace $\mathcal{H}_I$, with its dimension $D_I$ not vanishingly small compared to the largest nonequilibrium dimension $D_{\neqm}$ associated with the macroscopic observables of interest. The aperiodicity of a subspace implies the aperiodicity of each pure state within it, but also implies that at any given instant, almost all constituent pure states have a negligible overlap with this large subspace. This allows the aperiodicity of a subspace to constrain macroscopic thermalization of almost all constituent pure states within a finite time interval, where the aperiodicity of each constituent pure state alone would not suffice as seen in Corollary~\ref{cor:purestatethermalization}. Mathematically, this is possible because of the large dimension $D_I$  of the subspace (assumed non-vanishing compared to $D_{\neqm}$) that multiplies the aperiodicity in Eq.~\eqref{eq:equilibrationfractionfromaperiodicity} of Theorem~\ref{thm:aperiodicityimpliesmacroscopicthermalization}, whereas $D_I = 1$ for a pure state.

Further, while Theorem~\ref{thm:aperiodicityimpliesmacroscopicthermalization} constrains the average time spent in equilibrium for any given macroscopic observable, one is also interested in showing the simultaneous equilibration of several macroscopic observables. It is convenient to collect all observables of interest, assumed to be finite in number (as one can imagine that a practical measurement or observation cannot address more than a finite number), in a two dimensional array
\begin{equation}
    \hat{\mathcal{M}} = \left\lbrace \hat{A}_{m\ell}:\ m \in \mathbb{N}_M,\ \ell \in \mathbb{N}_{L_m},\ \left([\hat{A}_{m\ell}, \hat{A}_{nk}] = 0 \iff m = n\right)\right \rbrace
    \label{eq:observablearraydef}
\end{equation}
based on their commutation relations. In particular, we have $M$ different subsets of observables (with the $m$-th subset containing $L_m$ observables), each mutually commuting within such a subset, but not commuting across subsets. Let $\mathcal{M}_{\eqm}$ be a corresponding array of equilibrium values $A_{\eqm, m\ell}$. For each commuting subset, due to the shared eigenbases of the observables involved, we can define the joint equilibrium subspace as the intersection of those of each member observable (i.e. going back to Definition~\ref{def:equilibriumsubspace}, we want the subspace spanned by eigenvectors whose eigenvalues for all observables are within some $\epsilon$ of the respective equilibrium values):
\begin{equation}
    \mathcal{H}_{\eqm,m}(\hat{\mathcal{M}}, \mathcal{M}_{\eqm}, \epsilon) \equiv \bigcap_{\ell \in \mathbb{N}_{L_m}} \mathcal{H}_{\eqm}(\hat{A}_{m\ell}, A_{\eqm, m\ell}, \epsilon)
\end{equation}
of dimension $D_{\eqm}(\hat{\mathcal{M}}, \mathcal{M}_{\eqm}, \epsilon)$, and the corresponding nonequilibrium subspace $\mathcal{H}_{\neqm, m}$ of dimension
\begin{equation}
    D_{\neqm, m}(\hat{\mathcal{M}}, \mathcal{M}_{\eqm}, \epsilon) \equiv \dim \mathcal{H}_{\neqm,m}(\hat{\mathcal{M}}, \mathcal{M}_{\eqm}, \epsilon) \leq \sum_{\ell \in \mathbb{N}_{L_m}} D_{\neqm}(\hat{A}_{m\ell}, A_{\eqm,m\ell}, \epsilon).
\end{equation}
We can directly use this equilibrium subspace in Theorem~\ref{thm:aperiodicityimpliesmacroscopicthermalization} for each commuting subset of observables. However, between distinct commuting subsets, the equilibrium subspaces have more complicated overlaps in the Hilbert space, and there is no common \textit{subspace} of equilibrium (though non-subspace equilibrium subsets of the Hilbert space can be defined in a straightforward manner~\cite{GoldsteinetalMateMite1, GoldsteinetalMateMite2}). While one may argue for approximate commutation over macroscopic scales~\cite{vonNeumannThermalization, GoldsteinetalMacroscopicTypicality, GoldsteinetalMateMite1} or form a positive density by additively combining equilibrium projectors~\cite{TasakiTypicalityThermalization} to aggregate such equilibrium subspaces, we will find it convenient to keep these subspaces separate. It is also useful to define the maximum nonequilibrium dimension for $\hat{\mathcal{M}}$,
\begin{equation}
    D_{\neqm, \max}(\hat{\mathcal{M}}, \mathcal{M}_{\eqm}, \epsilon) \equiv \max_{m \in \mathbb{N}_M} D_{\neqm, m}(\hat{\mathcal{M}}, \mathcal{M}_{\eqm}, \epsilon),
    \label{eq:maxneqdimdef}
\end{equation}
and only retain the dependence on $\hat{\mathcal{M}}$ of the various equilibrium-subspace related quantities defined here, with the equilibrium values and resolution being understood to be implicit, e.g. $D_{\neqm, \max}(\hat{\mathcal{M}})$.

It is worth emphasizing the justification for considering noncommuting macroscopic observables. By the uncertainty principle~\cite{DiracQM, ShankarQM}, it is not possible to simultaneously measure any two noncommuting observables, e.g. $\hat{A}_{m\ell}$ and $\hat{A}_{nk}$ for $m \neq n$, to be close to unique values to \textit{arbitrary precision} due to their inherent spread in the state. However, especially in the macroscopic regime, such uncertainties may be negligibly small compared to the range of accessible values of these observables. It is therefore reasonable to expect that one can simultaneously measure all observables in $\hat{\mathcal{M}}$ to some finite resolution that is much larger than their expected uncertainty in most states of interest --- this is what happens regularly in observed statistical mechanics~\cite{vonNeumannThermalization, GoldsteinetalMacroscopicTypicality}.

Our specific technical goal here is to constrain the fraction of pure states in any basis in $\mathcal{H}_I$ that thermalize all such observables to a \textit{finite} resolution $\epsilon$ at at least a fraction $(1-\delta_{\tau})$ of times. To formalize this notion, we need a generalization of the equilibrium fraction in Eq.~\eqref{eq:equilibrationfractiondef} to our full array of observables:
\begin{equation}
            F_{\eqm}[\mathcal{B}_I, \hat{\mathcal{M}}](RT, \varepsilon) \equiv \frac{1}{D_I}\sum_{k=1}^{D_I}\int_{0}^{RT}\frac{\diff t}{RT} \ifc\left(\langle \psi_{Ik}(t)\rvert \proj_{\eqm,m}\lvert \psi_{Ik}(t)\rangle \geq 1-\varepsilon, \text{ for all } m \in \mathbb{N}_M\right).
            \label{eq:multiequilibrationfractiondef}
        \end{equation}
The conditional statement $\mathrm{X}$ inside $\ifc(\mathrm{X})$ here amounts to a logical $\mathrm{AND}$ of different conditional statements $\mathrm{X}_m$ for each $m$, satisfied by the intersection of the sets of states satisfying each of the latter. For sets $A, B \subseteq U$, using
\begin{equation}
    U \setminus (A \cap B) = \left(U \setminus A\right) \cup (U \setminus B),
\end{equation}
we have (with $|Y|$ denoting the cardinality or a measure of $Y$, and observing that the union may have common elements)
\begin{equation}
    |A \cap B| \geq |U| - (|U|-|A|) - (|U|-|B|).
\end{equation}
Therefore, extending this to the intersection of multiple sets, we get the bound
\begin{equation}
    F_{\eqm}[\mathcal{B}_I, \hat{\mathcal{M}}](RT, \varepsilon)  \geq 1-\sum_{m\in\mathbb{N}_M} \left(1-f_{\eqm}[\mathcal{B}_I, \hat{\mathcal{M}}](RT, \varepsilon) \right).
    \label{eq:multiequilibrationtosingleequilibrationsum}
\end{equation}
Using Eq.~\eqref{eq:equilibrationfractionfromaperiodicity} of Theorem~\ref{thm:aperiodicityimpliesmacroscopicthermalization} in Eq.~\eqref{eq:multiequilibrationtosingleequilibrationsum} and subsequently Eq.~\eqref{eq:maxneqdimdef} termwise, we get:
\begin{equation}
     F_{\eqm}[\mathcal{B}_I, \hat{\mathcal{M}}](RT, \varepsilon)\geq \left[1-\frac{1}{\varepsilon}\sqrt{\frac{M^2 D_{\neqm, \max}(\hat{\mathcal{M}})}{D_I \mathcal{A}_{[0,T]}[\mathcal{H}_I]}}\right].
     \label{eq:multiequilibrationfractionbound}
\end{equation}
We note that the required degree of aperiodicity for this to be close to $1$ scales as $\mathcal{A} \propto M^2$ for $M$ noncommuting subsets, but only as $\mathcal{A}\propto L_m$ for the subset with the largest nonequilibrium dimension. Therefore, ensuring that different noncommuting subsets indeed do not commute produces a better bound than if one splits a commuting subset further into multiple independent subsets.

Denoting the number of basis states that attain thermal equilibrium for all observables in $\hat{\mathcal{M}}$ for at least a fraction $(1-\delta_{\tau}) \in [0,1)$ of times in $[0,RT]$ by:
\begin{align}
    n_{\eqm}[\mathcal{B}_I, &\hat{\mathcal{M}}](RT, \varepsilon, \delta_{\tau}) \nonumber \\
    &\equiv \sum_{k=1}^{D_I} \ifc\left[\int_{0}^{RT}\frac{\diff t}{RT} \ifc\left(\langle \psi_{Ik}(t)\rvert \proj_{\eqm,m}\lvert \psi_{Ik}(t)\rangle \geq 1-\varepsilon, \text{ for all } m \in \mathbb{N}_M\right) \geq (1-\delta_\tau)\right],
    \label{eq:basismultiequilibrationfraction_def}
\end{align}
then Markov's inequality~\cite{RossProbability} applied to Eq.~\eqref{eq:multiequilibrationfractiondef} gives:
\begin{equation}
    \frac{n_{\eqm}[\mathcal{B}_I, \hat{\mathcal{M}}](RT, \varepsilon, \delta_{\tau})}{D_I}  \geq 1- \frac{1-F_{\eqm}[\mathcal{B}_I, \hat{\mathcal{M}}](RT, \varepsilon)}{\delta_{\tau}} = \frac{F_{\eqm}[\mathcal{B}_I, \hat{\mathcal{M}}](RT, \varepsilon)+\delta_{\tau}-1}{\delta_{\tau}}.
    \label{eq:multibasisequilibrationfractionbound}
\end{equation}
We can therefore formulate our finite-time thermalization statement, by combining Eq.~\eqref{eq:multibasisequilibrationfractionbound} with Eq.~\eqref{eq:multiequilibrationfractionbound}, as:
\begin{corollary}[Finite-time macroscopic thermalization of all orthonormal bases for an initial subspace]
\label{cor:multipleobservablefinitetimethermalization}
Let $\hat{\mathcal{M}}$ be an array formed from $M$ subsets each comprising of mutually commuting observables for a system with Hilbert space $\mathcal{H}$, such that the subsets do not (necessarily) commute with each other, as in Eq.~\eqref{eq:observablearraydef}.
For a subspace $\mathcal{H}_I$ of initial states of dimension $D_I$, showing aperiodicity $\mathcal{A}_{[0,T]}[\mathcal{H}_I]$ in the time interval $[0,T]$, for any orthonormal basis $\mathcal{B}_I$ for $\mathcal{H}_I$, the number of basis vectors in $\mathcal{B}_I$ for which all observables in $\hat{\mathcal{M}}$ attain thermal equilibrium to resolution $\epsilon$ for at least a fraction $(1-\delta_{\tau})$ of times in the interval $[0,RT]$ for any $R \in \mathbb{N}$, as defined in Eq.~\eqref{eq:basismultiequilibrationfraction_def}, is constrained as a fraction of $D_I$ by:
    \begin{equation}
        \frac{n_{\eqm}[\mathcal{B}_I, \hat{\mathcal{M}}](RT, \varepsilon, \delta_{\tau})}{D_I} \geq \left[1-\frac{1}{\varepsilon \delta_{\tau}}\sqrt{\frac{M^2 D_{\neqm, \max}(\hat{\mathcal{M}})}{D_I \mathcal{A}_{[0,T]}[\mathcal{H}_I]}}\right],
    \end{equation}
    which requires only a finite degree of aperiodicity over finite $T$ (for finite $\varepsilon$, $\delta_{\tau}$) to constrain this fraction to be arbitrarily close to $1$ if $D_I = \Theta(D_{\neqm,,\max}(\hat{\mathcal{M}}))$ in the thermodynamic limit, while $M$ remains finite.
\end{corollary}

\section{Discussion}
\label{sec:discussion}

\subsection{Thermalization by forgetfulness: a unifying mechanism?}

Aperiodicity (Definition~\ref{def:aperiodicity}), our proposed operational probe of macroscopic thermalization via Theorem~\ref{thm:aperiodicityimpliesmacroscopicthermalization} and especially Corollary~\ref{cor:multipleobservablefinitetimethermalization}, is essentially a measure of how much a subspace of initial states forgets itself over a finite amount of time. This simple fact of forgetfulness appears to have several intriguing connections: to seemingly very different mechanisms of microscopic thermalization (summarized in Table~\ref{tab:operationalstatmech}), to harnessing dissipation as a resource for computability~\cite{DissipativeQuantumComputation}, and to the fundamental problem of spectral statistics associated with ``quantum chaos''~\cite{Haake}. Elaborating on these connections is the subject of this section.


\subsubsection{Probes, mechanisms, and predictions}
\label{sec:disc_probesmechanismspredictions}

It is worth exploring what appears to be a common predictive structure that appears to emerge in our approach to operational quantum statistical mechanics (Sec.~\ref{sec:operationalquantumstatmech}), as developed here and in Refs.~\cite{dynamicalqthermalization, dynamicalpurestatethermalization}, across the different mechanisms in Table~\ref{tab:operationalstatmech}.
A general template of implication chains in which we can attempt to place the results of Refs.~\cite{dynamicalqthermalization, dynamicalpurestatethermalization} and the present work appears to be:
\begin{align}
\begin{split}
    \text{Computable probe behavior} \iff \text{Invariant } &\text{mechanism} \implies \text{A form of finite-time thermalization.} \\
    &\rotatebox[origin=c]{270}{$\implies$} \\
    \text{Energy eigenstate } &\text{structure} \implies \text{A form of infinite-time thermalization.} 
    \end{split}
     \label{eq:thermalizationimplicationtemplate}
\end{align}
All of these approaches allow completely filling in the first row. The second row is usually known from standard results, and some of these approaches provide a direct implication from the first to the second row. While the connection to the appropriate forms of eigenstate structure follows in a quantitatively similar manner in both the present work and Ref.~\cite{dynamicalqthermalization} (the latter revisited from a different eigenstate-independent perspective in \cite{dynamicalpurestatethermalization}), which have some classical counterparts, the structure that may hold information about long time implications appears to be more elaborate for quantum thermalization without classical counterparts~\cite{dynamicalpurestatethermalization}, though still based on time-translation invariance, and therefore leads to a somewhat different connection to eigenstates as discussed below. The ``invariant mechanism'' itself can be dropped for all practical purposes of computation, being an intermediary rather than an endpoint of the implication chain; its main role is to anchor the connection between the probe and the thermalization prediction, at minimum for intuition, and perhaps act as a generator for future connections and implications. The implication from the mechanism to the finite-time thermalization prediction is one-way in the sense that the thermalization of a given complete orthonormal basis of states (for observable-dependent criteria) or a given observable (for state-dependent ones) does not \textit{necessarily} require the mechanism. However, the mechanism does become \textit{necessary} if respectively demanding the thermalization of all bases~\cite{dynamicalqthermalization, dynamicalpurestatethermalization} or all conceivable observables (see the discussion preceding Theorem~\ref{thm:aperiodicityimpliesmacroscopicthermalization}) compatible with the setup in the Hilbert space.

Now, it is worth examining how our results fit into the template \eqref{eq:thermalizationimplicationtemplate}. As in Table~\ref{tab:operationalstatmech}, these implication chains are either state- or observable-dependent, but generally provide a clean separation between an input probe (correlator or return probability), an invariant mechanism, and an output prediction (different forms of thermalization across all observables and/or almost all states).

    \paragraph{\small Macroscopic thermalization:}\hspace{-1em} Our present results give the state-dependent implication chain:
    \begin{align}
    \begin{split}
    \text{Aperiodicity} \xLeftrightarrow[\text{Lemma~\ref{lem:ergodicfraction}}]{\text{Theorem~\ref{thm:aperiodicityimpliesmacroscopicthermalization}(1)}} \text{Partial } &\text{ergodicity} \xRightarrow[\text{Corollary~\ref{cor:multipleobservablefinitetimethermalization}}]{\text{Theorem~\ref{thm:aperiodicityimpliesmacroscopicthermalization}(2)}} \text{Finite-time macroscopic thermalization.}  \\
     &\rotatebox[origin=c]{270}{$\xRightarrow{\hphantom{Text}}$} \text{\ssmall{Corollary~\ref{cor:energyshelldelocalization}}} \\
    \text{Energy-shell/eigenbasis } &\text{delocalization} \xRightarrow[\text{\cite{Tasaki1998, ReimannRealistic, LindenetalEqb, ShortEqb, ShortDegenerate}}]{\text{\cite{TasakiTypicalityThermalization}}} \substack{\text{\small Infinite-time macroscopic thermalization} \\ \text{\small (or infinite-time universal equilibration).}} 
    \end{split}
    \label{eq:macroscopicimplicationchain}
    \end{align}
    Here, as with subsequent results, the thermalization implication is for almost all states in \textit{every} basis (of the relevant subspace or the full Hilbert space, depending on context), stronger than typicality (see also the discussion in Sec.~\ref{sec:operationalquantumstatmech}) and at almost all times (with the ``almost all'' qualifiers left implicit for brevity). Note that even without the second line involving energy eigenstate structure, the ``finite-time'' thermalization implication in the first line includes the extrapolation to infinite times due to the time-translation invariance of the Hamiltonian, as explicit in the statements of Theorem~\ref{thm:aperiodicityimpliesmacroscopicthermalization} and Corollary~\ref{cor:multipleobservablefinitetimethermalization}.
    \paragraph{\small Microscopic thermalization with statistical averages:}\hspace{-1em} Here, Ref.~\cite{dynamicalqthermalization} develops the following observable-dependent chain (adapted from ~\cite{dynamicalpurestatethermalization}), based on the behavior of an observable's autocorrelator, for predicting thermalization with (ergodic) time averages and (mixing) state averages:
    \begin{align}
    \begin{split}
    \text{Autocorrelator decay} \iff \text{Energy-band } &\text{thermalization} \implies \text{Finite-time averaged thermalization.}  \\
    &\rotatebox[origin=c]{270}{$\implies$} \\
    \text{Eigenstate } &\text{thermalization} \implies \text{Infinite-time averaged thermalization.} 
    \end{split}
    \label{eq:microscopicaveragesimplicationchain} 
    \end{align}
    Once again, the ``finite-time'' thermalization implication the first line directly allows arbitrarily long times, and therefore carries direct infinite-time implications even without requiring the bridge to eigenstate structure in the second line.
    \paragraph{\small Microscopic quantum thermalization:}\hspace{-1em} Here, the relevant criterion~\cite{dynamicalpurestatethermalization} is again largely observable-independent, but with some state dependence. One considers the behavior of an out-of-time-ordered correlator (OTOC) $G_{4}(t) = \Tr[\proj_A\proj_{\rho}(t)\proj_A\proj_{\rho}(t)]/\Tr[\proj_{\rho}]$ between an eigenspace $\proj_A$ of a few-body observable $\hat{A}$, and a subspace $\proj_{\rho}$ of all states that reduce to the same state on a finite number of qubits, which captures additional quantum properties compared to the two point correlator $G_{2}(t) = \Tr[\proj_A \proj_{\rho}(t)]/\Tr[\proj_{\rho}]$. If the former, which universally satisfies $G_4(t) \geq G_2^2(t)$, decays to near its minimum and factorize into the square of the latter, i.e. $G_4(t_1) \approx G_2^2(t_1)$ at any time $t_1$, this is shown~\cite{dynamicalpurestatethermalization} to correspond to an alignment of the eigenspaces of the observable with the subspace of initial states implying instantaneous quantum thermalization at $t_1$ in almost all states in the subspace. The implication chain at a \textit{single instant} of time is therefore:
    \begin{align}
    \begin{split}
    \text{OTOC factorization} \iff \text{Hilbert subspace } &\text{alignment} \implies \text{Instantaneous finite-time thermalization.}
    \end{split}
    \label{eq:microscopicquantumimplicationchain0}
    \end{align}
    OTOC factorization at specific times may be computed with finite resources in contrast to a direct computation of quantum thermalization in almost all initial pure states (even in a subspace of interest), making the above implication powerful for establishing quantum thermalization in almost all pure states at a given instant of time. Note that the thermalization of $G_2(t)$ itself is an input here, but can be established for almost all times by ``cloning'' Eq.~\eqref{eq:microscopicaveragesimplicationchain} into a ``doubled'' Hilbert space~\cite{dynamicalqthermalization}. Eq.~\eqref{eq:microscopicquantumimplicationchain0} may be regarded as the most direct version of the first line of Eq.~\eqref{eq:thermalizationimplicationtemplate}.
    
    For extrapolating quantum thermalization via $G_4(t)$ to other times, we argue in Ref.~\cite{dynamicalpurestatethermalization} that for a given finite range of times $\mathcal{T}$ and dimension $D_{\rho} = \Tr[\proj_{\rho}]$ for an observable $\proj_A$, there exists a subspace $\proj_{\rho}^{\max}$ such that the factorization of the corresponding OTOC (called ``Max OTOC'' below) at almost all times in $\mathcal{T}$ implies the factorization of $G_4(t)$ for every $\proj_{\rho}$ of the same dimension $D_{\rho}$ at almost all times in $\mathcal{T}$ \textit{and} longer time intervals generated by translating $\mathcal{T}$ (by time translation invariance as used in Lemma~\ref{lem:finitetimeergodicitytoinfinitetimethermalization} here), and therefore the quantum thermalization of almost all states at almost all times (``almost all'' again left implicit):
    \begin{align}
    \begin{split}
    \text{Max OTOC factorization in $\mathcal{T}$} \iff \text{Subspaces } &\text{align in $\mathcal{T}$} \implies \text{Quantum thermalization for $t \in \mathcal{T}$.} \\
    &\rotatebox[origin=c]{270}{$\xRightarrow{\hphantom{Text}}$} \text{\ssmall{Time-translation invariance}}  \\
    \text{Subspaces align in $\mathcal{T}+t_0$ for } & \text{any $t_0$} \implies \text{Quantum thermalization for $t \in \mathbb{R}$.}
    \end{split}
    \label{eq:microscopicquantumimplicationchain}
    \end{align}
    However, a convenient operational means to access this ``Max OTOC'', as well as the nontriviality of its behavior in generic cases, remains to be established. While we have not explicitly written out an eigenstate implication in Eq.~\eqref{eq:microscopicquantumimplicationchain}, as the intermediate implication concerns the alignment of \textit{all subspaces} $\proj_{\rho}$ of the same dimension as $\proj_{\rho}^{\max}$, the eigenstate implication readily follows by choosing $\proj_{\rho} = \proj_{\mathcal{E}}$, the projector onto a subspace $\mathcal{H}_{\mathcal{E}}$ spanned by any subset of energy levels $\mathcal{E} = \lbrace E_n\rbrace$ with $|\mathcal{E}| = D_{\rho}$ (noting that the corresponding OTOCs have no dynamics, therefore any $t$ can be chosen). Note that this is not restricted to energy shells of consecutive levels like energy-shell delocalization (as in Corollary~\ref{cor:energyshelldelocalization}), but can be any set of suitable size scattered over the spectrum. Intuitively, the alignment of such energy subspaces with the observable $\proj_A$ (given energy-band thermalization) should imply the smallness of almost all off-diagonal matrix elements of $\proj_A$ in the energy eigenbasis (via the quantum thermalization of almost all states in every basis of $\proj_\mathcal{E}$), in a stronger sense than what follows from (cloned) energy-band thermalization. Such a statement is associated with a weaker version of the off-diagonal part of the eigenstate thermalization hypothesis~\cite{srednicki1999eth, DAlessio2016}.

Interestingly, the ``computable probe behavior'' in all three cases appears to be a computable form of memorylessness (autocorrelators or return probabilities). In \eqref{eq:macroscopicimplicationchain}, aperiodicity measures the tendency of the initial subspace of states to forget itself over a finite interval of time. In \eqref{eq:microscopicaveragesimplicationchain}, the decay of the autocorrelator measures the tendency of an observable to lose correlations with its initial value across the Hilbert space, similar to Khinchin's ``molecular chaos'' mechanism for classical systems~\cite{KhinchinStatMech} (Eq.~\eqref{eq:KhinchinMolecularChaos}). In \eqref{eq:microscopicquantumimplicationchain0} and \eqref{eq:microscopicquantumimplicationchain}, the factorization of the OTOC is a more subtle, distinctly quantum form of memorylessness: from the Hilbert subspace alignment mechanism discussed in Ref.~\cite{dynamicalpurestatethermalization}, it follows that this can be regarded as the eigenspace of the observable losing correlations with any particular directions in the subspace of states, therefore treating the entire subspace uniformly. Equivalently, it can be viewed as the loss of memory of the subspace corresponding to $\proj_A$ under the completely positive~\cite{NielsenChuang} ``echo dynamics''~\cite{LoschmidtEcho1, LoschmidtEcho2, LoschmidtEcho3} $\proj_{\rho}^{\mathcal{K}}(t) = \kraus{1}{t} \proj_{\rho} \krausd{1}{t}$ with the Kraus operator $\kraus{1}{t} := \uh^\dagger(t) \proj_A \uh(t)$, where $\uh(t) = e^{-i\hat{H} t}$, under which $G_4(t) = \Tr[\proj_{\rho}^{\mathcal{K}}(t)\proj_{\rho}]/D_{\rho}$ is an autocorrelator of this subspace.

Further, the ``invariant mechanisms'' involve an alignment of states or observables: directly for the case of microscopic quantum thermalization~\eqref{eq:microscopicquantumimplicationchain0}; as an alignment of the observable with energy bands induced by the Hamiltonian in the space of operators for energy-band thermalization~\cite{dynamicalqthermalization} in the case of microscopic thermalization with averages~\eqref{eq:microscopicaveragesimplicationchain}; and in our present case of macroscopic thermalization~\eqref{eq:macroscopicimplicationchain}, partial ergodicity refers to the alignment of the history of a subspace of states with respect to itself over time, i.e. whether the dynamics generates sufficiently mutually orthogonal states to explore the Hilbert space~\cite{dynamicalqentanglement, dynamicalqergodicity}. We find these parallels interesting especially because the technical mechanisms underlying each implication chain appear, on the surface, to be formally different (except for the implication from the mechanism to the energy eigenstate row in both \eqref{eq:macroscopicimplicationchain} and \eqref{eq:microscopicaveragesimplicationchain}, for which both Corollary~\ref{cor:energyshelldelocalization} and Ref.~\cite{dynamicalqthermalization} respectively use essentially the same technique). It would be interesting to explore if they are more related than is presently apparent.


\subsubsection{Computability of forgetfulness}
\label{sec:disc_computability}

Having structured our result as a computational mode of prediction similar to other comparable results in Sec.~\ref{sec:disc_probesmechanismspredictions}, it is important to examine if the required input data is plausibly computable with finite resources. For this purpose, in the case of microscopic thermalization, Refs.~\cite{dynamicalqthermalization, dynamicalpurestatethermalization} sketched concrete measurement protocols compatible with quantum simulators~\cite{QuantumSimulationReview2024} to show the formal computability of their respective probes of forgetfulness with finite resources.

In our present case of macroscopic thermalization, the relevant probe of forgetfulness (aperiodicity) is a return probability of a subspace of states $\mathcal{H}_I$, which is one of the most straightforward objects to measure \textit{in principle}~\cite{NielsenChuang} in a quantum system given access to the projector $\proj_I$: prepare the system in the density operator $\hat{\rho}_I = \proj_I/D_I$, evolve it in time, and make a projective measurement onto $\proj_I$. Because of the anticipated small size of $\mathcal{H}_I$ relative to $\mathcal{H}$, i.e. we want $D_I \sim \Theta(D_{\neqm})$ so that thermalization statements about almost all states in $\mathcal{H}_I$ are sharp enough to resolve states within the nonequilibrium region $\mathcal{H}_{\neqm}$, preparing the state $\hat{\rho}_I$ is not straightforward. The easiest quantum simulation protocol --- prepare the full system in a maximally mixed state (e.g. by entangling each qubit in parallel with a unique auxiliary qubit, which only requires a single parallel set of local operations~\cite{NielsenChuang}) and then measure $\proj_I$ to prepare $\hat{\rho}_I$ precisely --- has postselection probability $D_I/D = \Theta(D_{\neqm})/D \approx 0$ in large systems, and is therefore overwhelmingly unlikely to produce the intended state in almost all trials.

To mitigate this expense, it appears convenient to use dissipative processes to prepare a proxy $\hat{\rho}_J$ for the initial state. To illustrate this in a concrete setting, consider a macroscopic (non-conserved) charge density $\hat{Q}$ with eigenvalues $q_k \in [-1,1]$ (such as $\hat{S}_z(\mathcal{L})$ in Sec.~\ref{sec:concentratedobservables}), and say that the corresponding equilibrium subspace of interest is $\proj_{\eqm}(\hat{Q}, 0, \epsilon)$, spanned by eigenvalues no further than $\epsilon$ from the presumed thermal value of $0$. For our initial subspace, let us say that we are satisfied with a set of states that are within a finite distance $\gamma: 0 < \gamma \leq 1-\epsilon$ of the lowest eigenvalue $-1$:
\begin{equation}
    \mathcal{H}_I := \vectorspan \left\lbrace \lvert q_k\rangle: q_k \leq -1+\gamma\right\rbrace.
\end{equation}
If $\hat{Q}$ were a (suitably rescaled) Hamiltonian (density), our subspace of initial states would correspond to an extensive number of low-energy states near the ground state. The projector $\proj_I$ is straightforward to implement via a projective measurement~\cite{NielsenChuang} of $\hat{Q}$, selecting outcomes in $[-1,-1+\gamma]$. To prepare the proxy state $\hat{\rho}_J$ (whose purpose will be described shortly), we start from the maximally mixed state and apply a dissipative quantum operation $\mathcal{D}$ that decreases the expectation value of $\hat{Q}$:
\begin{equation}
    \Tr[\mathcal{D}(\hat{\rho}_0)\hat{Q}] < 0.
\end{equation}
This could be done, for example if $\hat{Q}$ is an additive charge density, by taking a much larger external system with its corresponding charge $\hat{Q}_{\text{ext}}$ in the ground state, and coupling them weakly via a (well-understood) specific\footnote{We note that such an interaction may be identified and its dissipative properties determined in very system-specific ways, without requiring our more general framework that aims to provide system-independent computational modes of predicting thermalization, as a way of characterizing general mechanisms of the latter. In other words, we only require a single dissipative interaction that satisfies our requirements per nonequilibrium subspace of interest, rather than needing to generally characterize the most general process that may prepare such a state.} closed interaction of sufficient complexity that conserves $\hat{Q} + \hat{Q}_{\text{ext}}$ and is known to be dissipative in the precise way we will shortly require; in this case, the charge in our system of interest spreads into the external system and decreases the local charge density within the system. We expect that such a dissipative ``cooling'' process is likely to produce a proxy state that is diagonal in the $\lvert q_k\rangle$ basis:
\begin{equation}
    \hat{\rho}_J = \mathcal{D}(\hat{\rho}_0) = \sum_k p_k \lvert q_k\rangle\langle q_k\rvert,
\end{equation}
with $p_k \in [0,1]$ such that $\sum_k p_k = 1$ and
\begin{equation}
    p_k \geq \frac{p_0}{D_I} \text{ for all } k: q_k \leq -1+\gamma.
\end{equation}
Most importantly, the dissipation will have to be engineered so that $p_0$ is finite, in which case the density operator has significant support within $\mathcal{H}_I$, i.e. $\Tr[\hat{\rho}_J \proj_I] \geq p_0$.

Taking it for granted that at least one such dissipative process can be found and reliably characterized in a given computational platform of interest, such a preparable proxy state is useful as the dynamics of $\hat{\rho}_J$ constrains the dynamics of $\hat{\rho}_I$ according to:
\begin{equation}
    \Tr[\hat{\rho}_J(t) \proj_I] \geq \sum_{\substack{k:\\ q_k \leq -1+\gamma}} p_k \langle q_k(t)\rvert \proj_I\lvert q_k(t)\rangle \geq \sum_{\substack{k:\\ q_k \leq -1+\gamma}} \frac{p_0}{D_I}\langle q_k(t)\rvert \proj_I\lvert q_k(t)\rangle = p_0 \Tr[\hat{\rho}_I(t) \proj_I],
\end{equation}
due to which we have the following bound on aperiodicity (from Definition~\ref{def:aperiodicity}):
\begin{equation}
    \mathcal{A}_{\mathcal{T}}[\mathcal{H}_I] \geq \frac{p_0}{\displaystyle{\int_{\mathcal{T}}\frac{\diff t_1}{|\mathcal{T}|}\int_{\mathcal{T}}\frac{\diff t_2}{|\mathcal{T}|}\ \Tr[\hat{\rho}_J(t_1-t_2) \proj_I]}}.
    \label{eq:aperiodicity_proxystate_bound}
\end{equation}
Thus, all we require is the dissipative preparation of a ``low-charge'' nonequilibrium mixed state having sufficient \textit{state-wise} overlap with our initial nonequilibrium subspace to rigorously constrain aperiodicity from its measurable dynamics. As an aside, preparing this mixed state in this manner via purification (entanglement with an external system~\cite{NielsenChuang}) and subsequent dissipation, even for classical dynamics, can help statistically estimate the true behavior of the ensemble (which exists as a pure state in an extended system) without considering questions of classically sampling states from the ensemble (see also \cite{dynamicalqthermalization}). We take this to be a reasonable argument for the computability of this probe, pending the determination of such dissipative processes in different contexts of interest.

Another consideration of relevance to quantum simulations is that the dynamics of states is strictly evaluated at a discrete set of times by a finite computation, requiring $\diff t$ to be a discrete measure and $\mathcal{T}$ to be a finite set in Eq.~\eqref{eq:aperiodicity_proxystate_bound}. In this case, Theorem~\ref{thm:aperiodicityimpliesmacroscopicthermalization} implies longer-time thermalization only with the same \textit{discrete} measure, without constraining the times in between. However, note that Lemma~\ref{lem:finitetimeergodicitytoinfinitetimethermalization} allows the implication of macroscopic thermalization to carry over for \textit{any} set of shifts ($S_R$) of $\mathcal{T}$ by arbitrary real numbers, rather than just multiples of the original interval. In particular, we can e.g. shift a regular set of times in $[0,T]$ of spacing $\Delta t$ by fractions of the time step to form a regular lattice of any spacing $\delta t = \Delta t/w$, where $w\in \mathbb{N}$, for which Lemma~\ref{lem:finitetimeergodicitytoinfinitetimethermalization} provides the same quantitative bound as the original lattice. Taking the $w\to\infty$ limit, it follows (from the continuity of unitary dynamics in time) that establishing aperiodicity over a large but finite number of discrete time samples is sufficient to establish macroscopic thermalization over continuous intervals of time, with the fractional length of the exceptional set of times being comparable to their fraction in the discrete case.
In contrast, in analytical calculations, we expect that it may generically be sufficiently straightforward to calculate the dynamics of $\hat{\rho}_I$ itself, or if not, a similar ``low-energy'' initial state $\hat{\rho}_J$ can be used, which can generally be treated using the techniques of nonequilibrium dynamics e.g. in field theories~\cite{KamenevNonEqBook, bergesNonEqQFT}.

\paragraph{Operational predictions of macroscopic thermalization for multiple observables}
We emphasize that while our initial family of states may be nonequilibrium states with respect to a \textit{single} charge density (as was convenient in this illustration), whose thermalization over some finite set of times is then an \textit{input} into our framework, what follows from Theorem~\ref{thm:aperiodicityimpliesmacroscopicthermalization} and Corollary~\ref{cor:multipleobservablefinitetimethermalization} is the macroscopic thermalization in almost all such initial states of \textit{each} sufficiently concentrated observable in the full Hilbert space of the system at almost all times, and \textit{every} finite set of such observables jointly at almost all times, in any finite but sufficiently long interval of time.

As an idealized heuristic scenario for what this might look like in practice: consider a finite volume (quantum) gas of electrically charged and magnetized (i.e. nonzero spin) particles. Take a sufficiently large ensemble of nonequilibrium initial states each with a net global charge and spin of zero (so that the equilibrium state has zero charge and spin density), but typically with macroscopic magnetization within a large region $\mathcal{R}$ i.e. with spins strongly aligned in a specific direction within the region~\cite{Reif, ReichlStatMech}. For example, this could be an ensemble where $\mathcal{R}$, corresponding to half (or any other fraction of) the volume, is typically noticeably magnetized in one direction. Let us say that we merely observe that this magnetized ensemble loses its \textit{total} magnetization in $\mathcal{R}$ across sufficiently long times under its closed dynamics, which does not \textit{a priori} restrict local magnetizations of smaller volumes from being large (which may be mutually unaligned so that the net magnetization in $\mathcal{R}$ remains low) at those times. Then any (finite) set of macroscopic observables, which may comprise of macroscopic local magnetization densities in smaller (or larger) volumes of sufficient size, electric charge densities, currents, pressure, number densities, and other such observables inside, outside or straddling $\mathcal{R}$, must also (jointly) equilibrate to their typical values at almost all times for almost all states in this family, whatever their initial configuration (to a resolution $\epsilon_A$ determined by the requirement that $D_{\neqm}(\hat{A}, A_{\eqm},\epsilon_A) \sim D_I$ for each such observable $\hat{A}$). In particular, even smaller magnetized domains cannot exist at macroscopic scales at almost all times, nor can smaller (macroscopic) concentrations of electric charges, currents, pressure or number densities, assuming (as is typical) that the statistically typical configuration (or the ``macroscopic thermal state'') is a homogeneous and isotropic one corresponding to zero macroscopic spin and charge density. This extrapolation to an array of macroscopic observables at almost all times from a finite-time observation of a single observable that characterizes the nonequilibrium ensemble is the main rigorous feature of our operational formalism.

In this scenario, the decay of the one observable property of total magnetization in $\mathcal{R}$ implies the aperiodicity of this ensemble (states with low magnetization in $\mathcal{R}$ cannot have a large overlap with higher $\mathcal{R}$-magnetization states), and functions as computable information that replaces the eigenstate delocalization structure of the initial state used for infinite-time predictions~\cite{Tasaki1998, ShortDegenerate, TasakiTypicalityThermalization}, while recovering the latter and in addition certifying macroscopic thermalization at finite times as in the implication structure~\eqref{eq:macroscopicimplicationchain}. It would be interesting to obtain quantitative estimates for such scenarios, and explore the practical computability of macroscopic thermalization in this framework.

\paragraph{Practical heuristics for a single observable} Despite the idealized description above concerning the extrapolation to multiple observables (and times) from a single observable (and set of times), some heuristic considerations, whose operational description remains open, may become important in practice where the continuous equilibration dynamics of a single observable itself is concerned.
\begin{enumerate}
    \item \textit{Initial state preparation}: The first is that while the protocol described above may be carried out (presumably, being based on the maximally mixed state and dissipative ``cooling'' of a single observable with a larger bath) with polynomial efficiency in quantum simulators, in truly macroscopically large systems even polynomial efficiency (or any resource requirement that isn't exactly $O(1)$ or perhaps at best logarithmic in $N$)
can be prohibitively expensive. Such a concern is unique to macroscopic thermalization: for operational microscopic thermalization~\cite{dynamicalqthermalization, dynamicalpurestatethermalization}, given local dynamics, information from a finite subsystem of interest can spread only to a finite number of sites~\cite{GongHamazakiNoneqbBounds, LucasReview} within the finite interval of time required for the operational framework, making it sufficient to simulate dynamics on this finite number of sites where all resources become $O(1)$. For the present macroscopic framework, if the required subspace of states ($\hat{\rho}_I$ or $\hat{\rho}_J$) cannot be created exactly, one may have to resort to estimating its dynamics via a few representative samples of states in an orthonormal basis in the subspace (as always required for estimating macroscopic ensemble averages in the classical case, such as the phase space averages of \cite{KhinchinStatMech}) or even a single pure state that is sufficiently ``typical'' to represent the dynamics of the full subspace~\cite{tumulka_CT, CanonicalTypicalityPSW} in the quantum case.

\item \textit{Equilibration cascades}: Another practical consideration is that in many cases, where our predictions apply to ``almost all states'' in a nonequilibrium subspace, any ensemble of ``almost all nonequilibrium states'' concentrates in the thermodynamic limit to the ``edge'' of the nonequilibrium region, just shy of equilibrium. This means that such a description primarily characterizes the equilibration of \textit{near-equilibrium} states, and may miss highly nonequilibrium states. Formally, once we have shown that almost all states in the nonequilibrium region $\mathcal{H}_{\neqm}(\epsilon_1)$ equilibrate to the resolution $\epsilon_1$ (for some observable of interest kept implicit), we may ask about the equilibration of a smaller subspace $\mathcal{H}_{\neqm}(\epsilon_2) \subset \mathcal{H}_{\neqm}(\epsilon_1)$ for some $\epsilon_2 > \epsilon_1$ that is e.g. further from equilibrium. Here, it may turn out that $D_{\neqm}(\epsilon_2) \ll D_{\neqm}(\epsilon_1)$, even with $D_{\neqm}(\epsilon_2)/D_{\neqm}(\epsilon_1) \to 0$ in the thermodynamic limit; this is the case, for example, for the standard concentration bounds in Eq.~\eqref{eq:densities_Dneqm_bound}. Therefore, the equilibration of $\mathcal{H}(\epsilon_1)$ gives us no useful information about the equilibration of its subspace $\mathcal{H}(\epsilon_2)$ in this case. Further, showing that almost all states in $\mathcal{H}(\epsilon_2)$ equilibrate to resolution $\epsilon_2$ does not automatically imply that they continue onward to equilibrate to resolution $\epsilon_1$, because over any finite time $\mathcal{H}_{\neqm}(\epsilon_2)$ can only explore a negligible fraction of $\mathcal{H}_{\neqm}(\epsilon_1)$. While this situation is still improved compared to the classical situation~\cite{KhinchinStatMech} discussed in Sec.~\ref{sec:introconcentration} in that we can rigorously predict the finite-time equilibration of nonequilibrium states at all (particularly in the behavior of other observables as discussed above), it is worth examining how to address this remaining smaller problem of concentration that has a similar flavor to the classical one. We can still extract dynamical predictions for a single observable via a hypothesis of \textit{equilibration cascades}, which is essentially another typicality assumption, where we assume that once a state exits $\mathcal{H}_{\neqm}(\epsilon_2)$ into the larger $\mathcal{H}_{\neqm}(\epsilon_1)$, its dynamics follows that of a typical state in the latter (for the example in Eq.~\eqref{eq:spinobservable}, the observable has degenerate eigenspaces with minimum separations $\epsilon_1-\epsilon_2 \sim \Theta(1/L)$). Once again by assuming typicality, one may expect that sampling a finite number of such resolutions $\epsilon_j$ is sufficient to estimate the smooth thermalization profile of a single macroscopic observable. However, it remains to be characterized when these forms of typicality may hold, until which point such hypotheses remain heuristic with no rigorous predictive power, similar to earlier typicality approaches as discussed in Sec.~\ref{sec:operationalquantumstatmech} but in a narrower regime.

\item \textit{Hilbert space reduction}: Nevertheless, we may outline two possible qualitative directions that may help circumvent such remaining typicality hypotheses. The primary reason that such ``equilibration cascades'' become necessary is due to the macroscopic size of each nonequilibrium subspace \textit{relative} to a smaller one. While continuing to allow the nonequilibrium subspaces to be macroscopically large, two ways of limiting the relative size difference occur to us, while their precise formulations (if any) remain open. The first is that practical restrictions on dynamics (such as the smoothness of the dynamics of macroscopic observables) prevent a given nonequilibrium subspace from even in principle exploring all dimensions of a larger nonequilibrium subspace that contains it, where it becomes sufficient for a finite degree of aperiodicity to ensure exploration beyond the size of all \textit{dynamically accessible} nonequilibrium directions.
The second is to note that the overwhelming relative size between two nonequilibrium subspaces is primarily due to different microscopic configurations that are not accessible at macroscopic length scales, and such configurations can perhaps be eliminated via a spatial coarse graining procedure as in the renormalization group~\cite{KadanoffBlock, Kardar}. At the same time, we note that up to a point where such options can be developed, constraining the equilibration of a nonequilibrium subspace over finite times to single resolution $\epsilon_1$ for a single observable suffices to establish the equilibration of almost all states in it (even if concentrated to near-equilibrium) for all other observables to a resolution related to $\epsilon_1$ as described above.
\end{enumerate}


\subsubsection{The role of spectral statistics in thermalization}

The typicality and eigenstate-based layers of quantum statistical mechanics (discussed in Sec.~\ref{sec:operationalquantumstatmech}) have been historically intertwined with similar layers for quantum dynamical systems~\cite{Haake}. For the latter, the relevant observations are as follows. Haar-typical unitary time evolution operators (or their generating Hamiltonians) show specific spectral correlations in their energy eigenvalues $E_n$ that are collectively labeled ``Wigner-Dyson'' spectral statistics~\cite{Wigner1955, Dyson1962, DysonMehtaSR, Mehta}. Many measures of such correlations concentrate with respect to the Haar measure (i.e. have the same quantitative form in almost all systems), such as (normalized) nearest-neighbor level spacing statistics or longer-range spectral rigidity~\cite{Haake, Mehta}. A different typicality measure, especially for nearest-neighbor statistics, is provided by (largely) independent distributions of each energy level that shows no local correlations, leading to ``Poisson'' statistics~\cite{BerryTabor, Haake, Mehta}.

Beyond typicality, Wigner-Dyson statistics is often (but not always) found in the quantization of ``sufficiently complex'' (generally at least non-integrable or ergodic) classical dynamical systems~\cite{DKchaos, CGV, BerryStadium, BGS} that, for heuristically understood reasons~\cite{HOdA, BerrySpectralRigidity, Berry227, HaakePO, HaakePO2, Haake}, usually hew closer to Haar-typical behavior in the absence of various confounding factors (\cite{BGGSarithmetic, BraunHaakearithmetic} among many others, see e.g. Refs.~\cite{dynamicalqergodicity, dynamicalqentanglement, bakeranomalieserg} for a fuller discussion of confounding factors such as the (dis)appearance of accessible symmetries in the classical limit). Similarly, Poisson statistics is often (but not always) found in systems with an extensive number of \textit{conservation laws} e.g. generic (but not all) integrable systems whose standard quantizations contain several independent sequences of energy levels with different spacings, naturally leading to a statistical independence of nearest neighbors~\cite{BerryTabor, Haake, Mehta}. This apparent connection between quantized complex dynamics and Wigner-Dyson statistics on the one hand, and quantized integrable dynamics and Poisson statistics on the other, has suggested~\cite{srednicki1994eth, srednicki1999eth, DAlessio2016} a role of spectral statistics in thermalization in the same vein as the classical association between formal ergodic properties and thermalization (or formal integrability and the lack of thermalization) discussed in Sec.~\ref{sec:ergodicitythermalization}, though occasional exceptions have been observed numerically~\cite{JensenShankarETH, MaganWu, VardiCohenWDwithoutEigenstateSpread}.

Having established a connection between spectral statistics and a narrowly defined form of ergodic exploration of the full Hilbert space (as witnessed by cyclic permutations) in Ref.~\cite{dynamicalqergodicity}, a closely related continuous measure of ergodicity in Ref.~\cite{dynamicalqentanglement} that maintains a connection to spectral statistics, and having directly connected a partial form of the latter to macroscopic thermalization in the present work, we believe these developments bring us reasonably close to allowing a formal understanding of the connection between spectral statistics and thermalization. While spectral statistics, being a fine property of energy levels, is not accessible in the thermodynamically large many-body systems relevant for statistical mechanics, it appears that the associated dynamical process(es) of ergodicity is a sufficiently substantial commonality for an analysis of their relation to be of fundamental interest. The goal of these concluding paragraphs is therefore to examine in what ways spectral statistics does and does not relate to statistical mechanics, according to the dynamical picture developed thus far.

\paragraph{Spectral statistics and macroscopic thermalization} Our present results (e.g. Theorem~\ref{thm:aperiodicityimpliesmacroscopicthermalization} and Corollary~\ref{cor:multipleobservablefinitetimethermalization}) show that the partial ergodicity of an ensemble of nonequilibrium states is sufficient to guarantee macroscopic thermalization at finite times. The connection to spectral statistics emerges instead when considering the dynamics of \textit{pure} states (supported around some finite energy) over the Heisenberg timescale $t_{\htime} = 2\pi\dos(E)$ comparable to the density of states $\dos(E)$ (around the same energy), as in Ref.~\cite{dynamicalqentanglement}. For definiteness, let us consider a system having $D$ energy levels with Haar-random Circular Unitary Ensemble (CUE)~\cite{Haake} eigenvalues $E_n^{\text{CUE}} \in [0,2\pi)$, and another with uncorrelated (Poisson) random energies $E_{n}^{\text{Poisson}} \in [0,2\pi)$ in the same interval, evolving over discrete integer steps of time $t \in \mathbb{Z}$. For this system, the Heisenberg time is $t_{\htime} = D$. Then for an unbiased initial state $\lvert \psi_0\rangle = D^{-1}\sum_n \lvert E_n\rangle$ (which has maximal effective dimension, $d_{\effdim} = D$, as per Sec.~\ref{sec:corollaries_delocalization}), the fraction of Hilbert space explored in each case over the first $D$ steps (with $w=1/D$ for $t\in \mathbb{Z}_D$ and $w=0$ outside) becomes, in the large-$D$ limit with a slight change of notation in the subscript from $w$ to the domain of averaging (see \cite{dynamicalqentanglement} for a derivation from spectral statistics):
\begin{equation}
    \mu_{t \in \mathbb{Z}_D}^{\text{CUE}}\left(\lvert \psi_0\rangle\langle \psi_0\rvert\right) \to \frac{3}{4},\ \text{ in comparison with }\ \mu_{t \in \mathbb{Z}_D}^{\text{Poisson}}\left(\lvert \psi_0\rangle\langle \psi_0\rvert\right) \to \frac{1}{2}.
    \label{eq:spectralergodicity}
\end{equation}
This marks a quantitative $\Theta(1)$ difference in how much of the Hilbert space a state explores over the Heisenberg time scale, between these two prominent classes of spectral statistics. This difference has implications, for example, for few-body quantum teleportation~\cite{dynamicalqentanglement} and nonperturbative models of quantum gravity~\cite{JTreconstruction}. However, this is a transient phenomenon; over much longer time scales, e.g. $t \to \infty$ (relative to the Heisenberg scale) corresponding to a uniform weight over all of $t \in \mathbb{Z}$, both kinds of spectral statistics explore the full Hilbert space equally well:
\begin{equation}
    \mu_{t\in\mathbb{Z}}^{\text{CUE}}\left(\lvert \psi_0\rangle\langle \psi_0\rvert\right) = \mu_{t\in\mathbb{Z}}^{\text{Poisson}}\left(\lvert \psi_0\rangle\langle \psi_0\rvert\right) \to 1.
\end{equation}
Moreover, stronger connections to spectral statistics (that are more cumbersome to express in our present setting), producing a \textit{divergent} difference between the two forms of spectral statistics in addition to being able to quantitatively estimate Wigner-Dyson spectral rigidity, are obtained with an extension to a complete orthonormal basis of pure states in the language of cyclic permutations (analogous to Sec.~\ref{sec:classicalpermutations}) in the Hilbert space~\cite{dynamicalqergodicity} or quantum action-angle variables~\cite{JTreconstruction}.

Our primary emphasis here is that spectral statistics predominantly controls the rate of ergodic exploration of the Hilbert space by pure states over the Heisenberg timescale $t = \Theta(t_{\htime})$, which is typically exponentially large in the thermodynamic $N\to\infty$ limit as $\dos(E) \sim \exp(N)$. For coarse-grained states over finite times, one instead has partial ergodicity and macroscopic thermalization. The combination of Theorem~\ref{thm:aperiodicityimpliesmacroscopicthermalization}(1) and Eq.~\eqref{eq:spectralergodicity} in fact directly place macroscopic thermalization and spectral statistics in very different regimes on a continuum characterized by partial ergodicity:
\begin{equation}
\begin{aligned}
    \text{Macroscopic thermalization regime: }&  \mu_w(\proj_I) \gtrsim D_{\neqm}/D, &\text{ over times } t = \Theta(1), \\
    \text{Spectral statistics regime: }& \mu_w(\lvert \psi_I\rangle\langle\psi_I\rvert) = \Theta(1), &\text{ over times } t = \Theta(D).
\end{aligned}
\label{eq:spectralstatistics_and_partialergodicity_differentscales}
\end{equation}
This suggests that spectral statistics and macroscopic quantum thermalization are directly associated with different exploration scales (Heisenberg-scale and non-equilibrium region scale, respectively) of the \textit{same} form of ergodic dynamics. 
However, neither appears necessary nor sufficient for the other, despite existing on the same continuum; indeed, this is similar to the classical schematic in Fig.~\ref{fig:partialergodicity}.

\paragraph{Spectral statistics and microscopic thermalization} While aperiodicity has been defined here (Definition~\ref{def:aperiodicity}) for a single initial subspace, it can be averaged over a complete set of orthogonal subspaces~\cite{dynamicalqergodicity}, in which case it is lower-bounded by the spectral form factor~\cite{Haake} (an invariant spectral feature that can be obtained as the average aperiodicity of specific pure state bases). Over the Heisenberg timescale, this again imposes quantitative constraints admitting Wigner-Dyson statistics~\cite{dynamicalqergodicity}. At earlier times however, the form factor is sensitive to longer-range spectral features, and can constrain the average aperiodicity of large subspaces corresponding to \textit{few-body} eigenspaces. This leads to direct constraints on the \textit{speed} of the information scrambling process associated with microscopic quantum thermalization~\cite{dynamicalqspeedlimit, dynamicalqfastscrambling}. Once again, it appears that spectral statistics and microscopic quantum thermalization are associated with different extremes of time scales (Heisenberg and finite, respectively) of the same process of aperiodicity, but without strong reasons to suspect any direct connection across these scales. Another difference in scales worth emphasizing is that macroscopic thermalization requires the aperiodicity of a small subspace, while spectral statistics and speed limits on scrambling average aperiodicity over the full Hilbert space. We therefore have:
\begin{equation}
\begin{aligned}
    \text{Macroscopic thermalization regime: }&  \mathcal{A}_{[0,T]}[\mathcal{H}_I] = \Theta(1),\ &\text{ for } \dim \mathcal{H}_I = \Theta(D_{\neqm}),\ &T = \Theta(1), \\
    \text{Microscopic speed limit regime: }&  \langle \mathcal{A}_{[0,T]}[\mathcal{H}_I]\rangle_{\mathcal{H}} = \Theta(1),\ &\text{ for } \dim \mathcal{H}_I = \Theta(D),\ &T = \Theta(1), \\
    \text{Spectral statistics regime: }& \langle \mathcal{A}_{[0,T]}[\mathcal{H}_I]\rangle_{\mathcal{H}} = \Theta(D),\ &\text{ for } \dim \mathcal{H}_I = 1,\ &T = \Theta(D).
\end{aligned}
\label{eq:spectralstatistics_and_aperiodicity_differentscales}
\end{equation}


\paragraph{Emergence of statistical ``universality'' at different scales} In conclusion, it appears to be just the one computable property of forgetfulness of a \textit{single} object of interest that rigorously predicts its statistical behavior with respect to almost all other objects, underlying the different forms of quantum dynamics, statistical mechanics, and spectral features considered here. Over accessible timescales, it provides the basis for an operational formulation of statistical mechanics that allows computational predictions of thermalization i.e. that observables and states \textit{predictably} show typical behavior, at both microscopic and macroscopic length scales for classical and quantum systems. Over the exponentially large Heisenberg time scale for quantum systems, it accesses spectral features such as eigenvalue statistics and the structure of the energy eigenstates with respect to initial states or observables, connecting to the spectral behavior of typical \textit{systems}. For observable-dependent notions, the relevant mechanism appears to be a form of uniform alignment with respect to states~\cite{dynamicalqthermalization, dynamicalpurestatethermalization}, and for state-dependent ones a form of ergodic dynamics~\cite[and this work]{dynamicalqentanglement, dynamicalqergodicity}. It is possible that these mechanisms themselves may be extreme variants of each other for large subspaces (few-body observable eigenspaces) and small subspaces (states) in the Hilbert space. Even with the corresponding mechanisms apparently having been identified, there appears to be no \textit{a priori} theoretical reason\footnote{Beyond eigenstate structure conclusions that follow from the assumed \textit{exact} time-translation invariance stemming from isolated Hamiltonian dynamics, which itself may be restricted to only early times (without persisting long enough to carry implications for eigenstates) in practice due to external perturbations if a system is not perfectly isolated.} to expect this to mean that any one of these processes implies or causes the other without additional assumptions, especially due to their significant separations in time scales, length scales, and ``quantization scales''. However, it is perhaps satisfying that these different ``universal'' forms of statistical typicality in complex systems appear to essentially emerge from the same class of dynamical processes realized at different scales, which may enable a systematic exploration of additional factors that may contribute to their potential interdependence across different systems.


\subsubsection*{Acknowledgments}
This work was supported by NIST, by the National Science Foundation under Grant Number 1734006 (Physics Frontier Center), and by the Heising-Simons Foundation under Grant 2024-4848.

\printbibliography

\end{document}